\documentclass{article}% Basic Springer Nature Reference Style/Chemistry Reference Style
\usepackage{graphicx}%
\usepackage{multirow}%
\usepackage{amsmath,amssymb,amsfonts}%
\usepackage{amsthm}%
\usepackage{mathrsfs}%
\usepackage[title]{appendix}%
\usepackage{xcolor}%
\usepackage{textcomp}%
\usepackage{manyfoot}%
\usepackage{booktabs}%
\usepackage{authblk}
\usepackage{fullpage}
\usepackage{algpseudocode}%
\usepackage{listings}%
\theoremstyle{thmstyleone}%
\newtheorem{theorem}{Theorem}
\newtheorem{corollary}{Corollary}

\theoremstyle{thmstyletwo}%
\newtheorem{example}{Example}%

\theoremstyle{thmstylethree}%
\newtheorem{definition}{Definition}%

\usepackage[utf8]{inputenc}
\usepackage[T1]{fontenc}

\usepackage{stmaryrd}
\usepackage{bbm}
\usepackage{paralist}
\usepackage{todonotes}
\usepackage{booktabs}

\usepackage{verbatim} % for multiline comments

\usepackage{tikz}
\usetikzlibrary{arrows,chains,fit}
\usetikzlibrary{automata,positioning}
\usepackage{hyperref}
\usepackage{proof}

\usepackage{fancyvrb}
\usepackage{listings}
\usepackage[linesnumbered,boxed]{algorithm2e}

\usepackage{subcaption}
\usepackage{oz_journal}
\usepackage{bm}
\usepackage{longtable}
\newlength\opcodesnippetwidth
\newcommand{\nonTerm}[1]{\ensuremath{\langle\mathit{#1}\rangle}}

\newcommand{\draftcomment}[3]{{\color{#1}[#2: #3]} 
  \PackageWarning{WARNING: Draft comments visible}{#2: #3}}
\newcommand{\pr}[1]{\marginpar{\draftcomment{blue}{PR}{#1}}}
\newcommand{\dg}[1]{\marginpar{\draftcomment{blue}{DG}{#1}}}

\newcommand{\mdg}[1]{\marginpar{\draftcomment{blue}{DG}{#1}}}
\newcommand{\ze}[1]{\marginpar{\draftcomment{blue}{ZE}{#1}}}
\newcommand{\cl}[1]{\marginpar{\draftcomment{blue}{CL}{#1}}}
\newcommand{\ja}[1]{\marginpar{\draftcomment{blue}{JA}{#1}} }

\newcommand{\set}[1]{\left\{ #1 \right\}}

\newcommand{\defeq}{\stackrel{\mathsf{def}}{=}}

\newcommand{\typeJudg}[2]{#1 : #2}

\newcommand{\instropdef}[1]{$\Omega_{#1} = (G_{\mathit{#1}}, R_{\mathit{#1}}, I_{\mathit{#1}})$}

\newcommand{\maxinstrop}{\Omega_{\mathit{max}}}

\newcommand{\suminstrop}{\Omega_{\mathit{sum}}}
\newcommand{\forallinstrop}{\Omega_{\forall,P}}

\newcommand{\ncaninstrop}[1]{\Omega_{\mathit{nc}}(#1)}
\newcommand{\caninstrop}[1]{\Omega_{\mathit{c}}(#1)}
\newcommand{\ncaninstropgen}{\ncaninstrop{(M, \circ, e), h}}
\newcommand{\caninstropgen}{\caninstrop{(M, \circ, e), h}}
\newcommand{\ncaninstropnoargs}{\Omega_{\mathit{nc}}}
\newcommand{\caninstropnoargs}{\Omega_{\mathit{c}}}

\newcommand{\figstop}{}

\newcommand{\eldarica}{\textsc{Eld\-arica}}
\newcommand{\tricera}{\textsc{Tri\-Cera}}
\newcommand{\cpachecker}{\textsc{CPA\-checker}}
\newcommand{\seahorn}{\textsc{Sea\-Horn}}
\newcommand{\jayhorn}{\textsc{Jay\-Horn}}
\newcommand{\rusthorn}{\textsc{Rust\-Horn}}

\newcommand{\monocera}{\textsc{Mono\-Cera}}
\newcommand{\veriabs}{\textsc{Veri\-Abs}}
\newcommand{\korn}{\textsc{Korn}}

\newcommand{\monoc}{\textsc{Mono}}
\newcommand{\tri}{\textsc{Tri}}
\newcommand{\sea}{\textsc{Sea}}
\newcommand{\cpa}{\textsc{CPA}}

\definecolor{mGreen}{rgb}{0,0.6,0}
\definecolor{mGray}{rgb}{0.5,0.5,0.5}
\definecolor{mPurple}{rgb}{0.58,0,0.82}
\definecolor{backgroundColour}{rgb}{0.95,0.95,0.92}
\definecolor{backgroundColourOP}{rgb}{0.97,0.97,0.94}

\lstdefinestyle{CStyle}{
    backgroundcolor=\color{backgroundColour},   
    commentstyle=\color{mGreen},
    keywordstyle=\color{magenta},
    numberstyle=\tiny\color{mGray},
    stringstyle=\color{mPurple},
    basicstyle=\ttfamily,
    breakatwhitespace=false,         
    breaklines=true,                 
    captionpos=b,                    
    keepspaces=true,                 
    numbers=left,                    
    numbersep=5pt,                  
    showspaces=false,                
    showstringspaces=false,
    showtabs=false,                  
    tabsize=2,
    language=C
}

\lstdefinelanguage{ExtWhile}{
    morekeywords={skip, if, else, while, for, assert, assume, const, store, select, aggregate, int},
    sensitive=false,
    morecomment=[l]{//},
    morecomment=[s]{/*}{*/},
}

\lstdefinestyle{ExtWhileStyle}{
    backgroundcolor=\color{backgroundColour},   
    commentstyle=\color{mGreen},
    keywordstyle=\color{magenta},
    numberstyle=\tiny\color{mGray},
    stringstyle=\color{mPurple},
    basicstyle=\ttfamily\small,
    breakatwhitespace=false,
    breaklines=true,
    captionpos=b,
    keepspaces=true,
    numbers=left,
    numbersep=5pt,
    showspaces=false,
    showstringspaces=false,
    showtabs=false,
    tabsize=2,
    language=ExtWhile
}

\lstdefinestyle{ExtWhileStyleOp}{
    backgroundcolor=\color{backgroundColourOP},   
    commentstyle=\color{mGreen},
    keywordstyle=\color{magenta},
    numberstyle=\tiny\color{mGray},
    stringstyle=\color{mPurple},
    basicstyle=\ttfamily\scriptsize,
    breakatwhitespace=false,
    breaklines=true,
    captionpos=b,
    keepspaces=true,
    numbers=left,
    numbersep=5pt,
    showspaces=false,
    showstringspaces=false,
    showtabs=false,
    tabsize=2,
    language=ExtWhile
}

\definecolor{nodeGray}{rgb}{0.96,0.96,0.96}
\definecolor{cornerGray}{rgb}{0.90,0.90,0.87}
\definecolor{beaublue}{rgb}{0.74, 0.83, 0.9}
\definecolor{mistyrose}{rgb}{1.0, 0.89, 0.88}
\tikzstyle{nodeone} = [rounded corners, text width=1.7cm, minimum height=1.8cm,text centered, draw=black, fill = nodeGray, font = \scriptsize]
\tikzstyle{nodetwo} = [rounded corners, text width=1.5cm, minimum height=2.7cm, minimum width = 5.5cm, text centered, draw=black, fill = beaublue, font = \scriptsize]
\tikzstyle{cornernode} = [rectangle, below right,draw, fill=cornerGray]

\tikzstyle{arrow} = [thick,->,>=stealth]

\newcommand{\assign}{\;\texttt{=}\;}
\newcommand{\eqeq}{\;\texttt{==}\;}
\newcommand{\semic}{\texttt{;}\;}
\newcommand{\add}{\;\texttt{+}\;}
\newcommand{\mult}{\;\texttt{*}\;}

\newcommand{\twopartdef}[3]
{
	\left\{
		\begin{array}{ll}
			#1 & \mbox{if } #2 \\
			#3 & \mbox{otherwise}
		\end{array}
	\right.
}
\newtheorem{lemma}{Lemma}%

\newif\ifusebattery
\usebatterytrue

\newif\ifuseimpl
\begin{document}

\title{A Program Instrumentation Framework for Automatic Verification
}
% \author{Some Authors}


% \authorrunning{Program Instrumentation for Automatic Verification}

\author[1,$\ddagger$]{Jesper Amilon}%\email {jamilon@kth.se}

\author[2,$\ddagger$]{Zafer Esen}%\email {zafer.esen@it.uu.se}

\author[1,$\ddagger$]{Dilian Gurov}%\email {dilian@kth.se}

\author[1,$\ddagger$]{Christian Lidstr\"om}%\email {clid@kth.se}

\author[2,3,$\ddagger$]{Philipp R\"ummer}%\email {philipp.ruemmer@it.uu.se}

\author[1,$\ddagger$]{Marten Voorberg}%\email {voorberg@kth.se}

\affil[1]{KTH Royal Institute of Technology, Stockholm, Sweden}

\affil[2]{Uppsala University, Uppsala, Sweden}

\affil[3]{University of Regensburg, Regensburg, Germany}

\affil[$\ddagger$]{These authors contributed equally to this work.}

% \institute{KTH Royal Institute of Technology, Stockholm, Sweden \and 
           % Uppsala University, Sweden \and University of Regensburg, Germany}

% \pagestyle{plain}

% \maketitle

% \keywords{Program Instrumentation, Automatic Verification, Ghost Code}

\maketitle

\abstract{
In deductive verification and software model checking,
dealing with certain specification language constructs can be problematic
when the back-end solver is not 
sufficiently powerful or lacks the required theories. 
One way to deal with this is to transform, for
verification purposes, the program to an equivalent one
not using the problematic constructs, and to reason about 
this equivalent program instead.
In this article, we propose \emph{program instrumentation} as a unifying verification paradigm that subsumes various existing ad-hoc approaches, has a clear formal correctness criterion, can be applied automatically, and can transfer back witnesses and counterexamples. 
%
% [Instrumentation = Instrumentation Operator + Application Strategy]
%
We illustrate our approach on the automated verification of programs that involve quantification and aggregation operations over arrays, such as the maximum value or sum of the elements in a given segment of the
array, which are known to be difficult to reason about automatically.
%
% We formalise array aggregation operations as monoid homomorphisms. 
% [Extend]
%
We implement our approach in the \monocera\ tool, which is tailored to the verification of programs with aggregation, 
and evaluate it on example programs, including SV-COMP programs.
% [Adapt]
}

%%%%%%%%%%%%%%%%%%%%%%%%%%%%%%%%%%%%%%%%%%%%%

%%%%%%%%%%%%%%%%%%%%%%%%%%%%%%%%%%%%%%%%%%%%%%%%%%%%%%%%%%%%%%%%
%%%%%%%  START OF TEXT %%%%%%%%%%%%%%%%%%%%%%%%%%%%%%%%%%%%%%%%
%%%%%%%%%%%%%%%%%%%%%%%%%%%%%%%%%%%%%%%%%%%%%%%%%%%%%%%%%%%%%%%%

\section{Introduction}
\label{sec:introduction}

Program specifications are often written in expressive, high-level languages: for instance, in temporal logic~\cite{DBLP:reference/mc/2018}, in first-order logic with quantifiers~\cite{DBLP:books/daglib/0022394}, in separation logic~\cite{DBLP:conf/lics/Reynolds02}, or in specification languages that provide extended quantifiers for computing the sum or maximum value of array elements~\cite{DBLP:books/daglib/p/LeavensBR99,acsl}. Specifications commonly also use a rich set of theories; for instance, specifications could be written using full Peano arithmetic, as opposed to bit-vectors or linear arithmetic used in the program. Rich specification languages make it possible to express intended program behaviour in a succinct form, and as a result reduce the likelihood of mistakes being introduced in specifications.

There is a gap, however, between the languages used in specifications and the input languages of automatic verification tools. Software model checkers, in particular, usually require specifications to be expressed using program assertions and Boolean program expressions, and do not directly support any of the more sophisticated language features mentioned. In fact, rich specification languages are challenging to handle in automatic verification, since satisfiability checks can become undecidable (i.e., it is no longer decidable whether assertion failures can occur on a program path), and techniques for inferring program invariants usually focus on simple specifications only.

To bridge this gap, it is common practice to \emph{encode} high-level specifications in the low-level assertion languages understood by the tools. For instance, temporal properties can be translated to B\"uchi automata, and added to programs using ghost variables and assertions~\cite{DBLP:reference/mc/2018}; quantified properties can be replaced with non-determinism, ghost variables, or loops~\cite{DBLP:conf/sas/BjornerMR13,DBLP:conf/sas/MonniauxG16}; sets used to specify the absence of data-races can be represented using non-deterministically initialized variables~\cite{cell2010}. By adding ghost variables and bespoke ghost code to programs~\cite{DBLP:journals/fmsd/FilliatreGP16}, many specifications can be made effectively checkable.

%Encodings for various kinds of specifications, often ingeniously designed, have been proposed in the literature, and mastering the art of writing suitable program assertions is a significant part of the 

The translation of specifications to assertions or ghost code is today largely designed, or even carried out, by hand. This is an error-prone process, and for complex specifications and programs it is hard to ensure that the low-level encoding of a specification faithfully models the original high-level properties to be checked. Mistakes have been found even in industrial, very carefully developed specifications~\cite{DBLP:conf/atva/PriyaZSVBG21},
% \pr{more examples?}
and can result in assertions that are vacuously satisfied by any program. Naturally, the manual translation of specifications also tends to be an ad-hoc process that does not easily generalise to other specifications.

This article proposes the first general framework to automate the translation of rich program specifications to simpler program assertions, using a process called \emph{instrumentation.} Our approach models the semantics of specific complex operations using program-independent \emph{instrumentation operators,} consisting of (manually designed) rewriting rules that define how the evaluation of the operator can be achieved using simpler program statements and ghost variables. The instrumentation approach is flexible enough to cover a wide range of different operators, including operators that are best handled by weaving their evaluation into the program to be analysed. While instrumentation operators are manually written, their application to programs can be performed in a fully automatic way by means of a search procedure. The soundness of an instrumentation operator can be shown formally, once and for all, by providing an \emph{instrumentation invariant} that ensures that the operator can never be used to show correctness of an incorrect program.
%For the operators presented in this paper, a machine-checked correctness proof was done using
%Frama-C~\cite{DBLP:conf/sefm/CuoqKKPSY12}.

% An extended technical report is to appear~\cite{techreport}, containing 
% Additional instrumentation operator definitions,
% correctness proofs,
% and detailed evaluation results
% can be found in the accompanying extended 
% report~\cite{techreport}.
%, which have been left out
% of the present paper for lack of space.

\paragraph{Comparison to CAV paper}

This article extends our previous work~\cite{amilon-et-al-cav23} by including new instrumentation operators, now covering the full set of extended quantifiers in the specification language ACSL~\cite{acsl}. Instrumentation operators for extended quantifiers are presented more systematically using monoid homomorphisms.
A more comprehensive correctness result is included, where \emph{completeness} in the typical sense is shown for certain classes of programs. Several existing sections of the CAV paper have been extended by adding more thorough explanations or discussions.

%% If the battery example should be used
% \ifusebattery 
    % \lstinputlisting[style=CStyle,basicstyle=\ttfamily\small,morekeywords={assert}]{examples/battery_diag_for.c}
\begin{figure}
\centering
\begin{minipage}{0.87\textwidth}
    % \resizebox{\textwidth}{!}{%
        % \centering
\begin{lstlisting}[style=CStyle,basicstyle=\ttfamily\small,morekeywords={assert}]
typedef struct {
    int batt_volt[10];
    int max_voltage;  //Maximum voltage value in pack
} BATTERY_PACK;

extern BATTERY_PACK nondet_BATTERY_PACK();
BATTERY_PACK bpack;

void calc_pack_max() {
  bpack.max_voltage = bpack.batt_volt[0];
  for (int i = 1; i < 10; i++)
    if (bpack.batt_volt[i] > bpack.max_voltage)
      bpack.max_voltage = bpack.batt_volt[i];
}
void main() {
  bpack = nondet_BATTERY_PACK(); //non-det initialization
  calc_pack_max();
  assert(bpack.max_voltage == \max(bpack.batt_volt, 0, 10));
}
\end{lstlisting}
            \caption{Example program calculating the maximum voltage in a battery pack consisting of an array of batteries\figstop{}}
            \label{fig:battery_example}
    % }
\end{minipage}
\end{figure}
% \todo{JA: Change 10 to N}
% \todo{JA: Should probably add the instrumented program for Fig. 1}

\subsection{Motivating Examples}\label{sec:motivating_examples}
We illustrate our approach on two simple examples.

\paragraph{Array Aggregation}
Consider first the program shown in \autoref{fig:battery_example}, which operates on a battery pack powering the electric engine of a vehicle. Although this particular implementation is written by us, it is conceptually close to (significantly more complicated) code we encountered in collaborations with automotive companies. The battery pack is represented as a C~struct, consisting of an array of batteries (representing the measured voltage at each battery) and a field for the maximum voltage value among the batteries. In the program, the \lstinline$calc_pack_max$ function calculates the maximum value by looping through the array, and in the main function, correctness of \lstinline$calc_pack_max$ is asserted by relying on the \lstinline$\max$ operator. Intuitively, \lstinline$\max(a, l, u)$ extracts the maximum value of the array \lstinline$a$ in the interval [\lstinline$l$, \lstinline$u$).

\begin{figure}[!tb]
    \centering
    \begin{minipage}{0.87\linewidth}
%\footnotesize
\begin{class}{\bm{\Omega}_{\mathit{max}} \textbf{
(Instrumentation operator)}}
\begin{schema}{\mathbf{G_{\mathit{max}}} \textbf{ (Ghost variables)}}
    \mathtt{ag\_lo},~\mathtt{ag\_hi},~
    \mathtt{ag\_max:~Int},
    \ \mathtt{ag\_ar:~Array\ Int}
    \where
    \mathtt{init(ag\_lo)} \,=\,
    \mathtt{init(ag\_hi)} \,=\, 0,\;
    \mathtt{init(ag\_max)} \,=\, -\infty,\;
    % \mathtt{init(ag\_max\_idx)} ~=~ 0,\\
    \mathtt{init(ag\_ar)} \,=\, \_
\end{schema}
\\ %[-22pt]
\begin{schema}{\mathbf{R_{\mathit{max}}} \textbf{ (Rewrite~rules)}}
\begin{array}{l@{\quad}l@{\quad}l@{\quad\qquad}l}
     \texttt{a' = store(a, i, x)}  &\leadsto&  & \\%\text{(R1)} \\
     \multicolumn{4}{l}{\hspace{2em}\mbox{
\begin{lstlisting}[linewidth=\opcodesnippetwidth, style=ExtWhileStyleOp]
if (i == ag_lo - 1 && ag_lo != ag_hi) {
    assert(ag_ar == a);
    ag_lo = ag_lo = i;
    ag_max = max(ag_max, x);
} else if (i == ag_hi && ag_lo != ag_hi) {
    assert(ag_ar == a);
    ag_hi = i;
    ag_max = max(ag_val, x);
} else {
    ag_max = -$\infty$;
    ag_lo = i;
    ag_hi = i + 1;
}
a' = store(a, i, x);
ag_ar = a'
\end{lstlisting}}}\\
     \texttt{x = select(a, i)}  &\leadsto&
        \text{code similar to rewrites of \texttt{store}} & \\%\text{(R2)}\\
     \texttt{r = \textbackslash max(a, l, u)}  &\leadsto&  & \\%\text{(R3)} \\
     \multicolumn{4}{l}{\hspace{2em}\mbox{
\begin{lstlisting}[linewidth=\opcodesnippetwidth, style=ExtWhileStyleOp]
if (u <= l) {
  r = $-\infty$;
} else {
  assert(ag_ar == a && l == ag_lo && u == ag_hi);
  r = ag_max;
}
\end{lstlisting}}}\\
\end{array}
\end{schema}
\\ %[-22pt]
\begin{schema}{\mathbf{I_{\mathit{max}}}\textbf{ (Instrumentation invariant)}}
    \mathtt{ag\_lo} = \mathtt{ag\_hi} \lor 
    \mathtt{ag\_max} = \mathtt{\backslash max (ag\_ar, ag\_lo, ag\_hi)}
     % \mathtt{ag\_lo} = \mathtt{ag\_hi} \vee 
     % % (\mathtt{ag\_lo} \leq \mathtt{ag\_max\_idx} < \mathtt{ag\_hi}
     % \: \wedge \\
     %        \qquad\qquad\qquad\qquad\;\;
     %       \mathtt{ag\_max} = \mathtt{\backslash max}(\mathtt{ag\_ar}, \mathtt{ag\_lo}, \mathtt{ag\_hi}) \: \wedge \\
     %       \qquad\qquad\qquad\qquad\;\;
     %      \mathtt{ag\_max} = \mathtt{ag\_ar[ag\_max\_idx]})
      % \end{align*}
\end{schema}
\end{class}

    \caption{Definition of the instrumentation operator $\maxinstrop{}$\figstop{}}
    \label{fig:max-instr-op}
    \end{minipage}
\end{figure}

Although simple in principle, the program in \autoref{fig:battery_example} is non-trivial to handle for automatic approaches. Software model checkers, in particular, tend to infer loop invariants for the program that explicitly refer to all array elements, resulting in high computational complexity. Hardness is increased further when considering larger arrays, or even unbounded arrays, where more compact representations of the loop invariant are needed. In the industry, programs calculating some aggregate measure over an array are ubiquitous, wherefore there is a clear need for verification tools to support aggregation operators like \lstinline$\max$.

To verify the program, we define the instrumentation operator $\maxinstrop{}$, shown in \autoref{fig:max-instr-op}. For sake of presentation, all instrumentation operators in this article are defined using functional arrays, using \lstinline$select$ and \lstinline$store$ to access and update the fields of an array, respectively. It is clear that the program in \autoref{fig:battery_example} can be rewritten to this style, but we leave this transformation implicit.
Like all instrumentation operators, $\maxinstrop{}$ is comprised of three parts, the \emph{ghost variables}, the \emph{rewrite rules}, and the \emph{instrumentation invariant}. While instrumentation operators for array properties are discussed at length in \autoref{sec:instrumentation_operators_for_arrays}, we summarise here the key features of $\maxinstrop{}$ in the context of \autoref{fig:battery_example}.

The essential function of $\maxinstrop{}$ is to induce, by applying the rewrite rules, a program transformation that mitigates the verification tasks. For this, assignments in a program that include the functions \lstinline$select$, \lstinline$store$, or \lstinline$\max$ can rewritten to more complex pieces of code that, in addition to accessing the array, also update the ghost variables of the instrumentation operator. The goal is to rewrite the program in such a way that its correctness can be checked without having to evaluate the problematic function~\lstinline$\max$; to this end, the instrumentation operator introduces ghost variables that track the maximum value of the array in an interval~$[\mathtt{ag\_lo}, \mathtt{ag\_hi})$.
The rewrite rules have to be applied to some selection of the read/write statements, so that the ghost variables are updated appropriately.
%
% transform the program, by adding ghost variables tracking relevant properties of the array, in such a way that the resulting program can be verified \emph{without having to reason about the array itself.} 
% Specifically, for \autoref{fig:battery_example}, the ghost variables track the maximum value and its index in an interval~$[\mathtt{ag\_lo}, \mathtt{ag\_hi})$ of the array, respectively. 
% Then, we rewrite any read/write statements to the array, such that they also update the ghost variables appropriately. 
% Our approach would then add ghost code (or \emph{instrumentation} code) at each array access that updates the ghost variables according to some given rules. 
%
%
For example, if the program reads from position~\lstinline{ag_hi} in the array, then the tracked interval~$[\mathtt{ag\_lo}, \mathtt{ag\_hi})$ can be extended by incrementing the ghost variable~\lstinline{ag_hi} for the upper interval bound, and similarly for the lower bound for accesses at position $\mathtt{ag\_lo}-1$. 
Moreover, if the value read is greater than the previously tracked maximum value~\lstinline$ag_max$, this ghost variable is updated as well.
%
% \todo[inline]{DG: Sentence about not showing here the result, but in the next example.}

Instrumentation operators do not, for a given program, produce a unique instrumentation, but instead give rise to a set of instrumented programs, which we call the \emph{instrumentation space}. Essentially, the instrumentation space is formed by considering all possible sets of statements to which the rewrite rules can be applied. 
In this example, we need to rewrite the array reads at lines~10 and~13. Thereafter, we also rewrite the \lstinline$\max$ aggregation in terms of the ghost variables. As manifested in our experiments (\autoref{sec:evaluation}), this shifting of verification domain, from array elements to ghost variables, significantly reduces the complexity for automatic verification tools to infer loop invariants, thus allowing programs such as this one to be verified. 

The soundness of the instrumentation operator is tied tightly to the instrumentation invariant. Instrumentation invariants are formulas that can (only) refer to the ghost variables introduced by an instrumentation operator, and are formulated in such a way that they hold \emph{in every reachable state of every program in the instrumentation space.} To maintain their invariants, instrumentation operators use shadow variables that duplicate the values of program variables. For $\maxinstrop{}$, the invariant is that, if the ghost variables track some non-empty interval, then the tracked maximum value is indeed the maximum value of the shadow ghost array~\lstinline$ag_ar$ in that interval. The connection between the shadow variable and the underlying program variable is established by letting the rewrite rules assert their equality.

% In the example, the ghost code to track the maximum value can be added at the array reads in lines~10, 12 and 13.\todo{JA: Add the instrumented program} When we enter function~\lstinline{calc_pack_max}, there is no information about the maximum value of the array; therefore the variables are initialized to $\mathtt{ag\_lo} = \mathtt{ag\_hi} = 0$. 
% At the first array read in line~10, the tracked interval is updated by assigning \lstinline{ag_max} to the read value and updating the bounds to encode the interval $[0, 1)$. Each loop iteration will update the bounds to $[0, i)$, and, whenever a new greatest-thus-far element is encountered, also the ghost variables tracking the maximum value and its index are updated. Thus, when exiting the loop, the ghost variables will encode the interval $[0, 10)$ and the maximum value of the array. Lastly, at the assert statement at line~18, the \lstinline$\max$ aggregation will be rewritten in terms of the ghost variables, where we also assert that they track the correct interval (in this case, $[0, 10)$). 
% \pr{show the rewritten program in the appendix?
% JA: This question still applies 2024-02-27}

\paragraph{Non-linear Arithmetic}
Our main focus are instrumentation operators for verification of properties over arrays. However, instrumentation operators also generalise to other verification tasks, as illustrated here for the purpose of verifying non-linear arithmetic properties. 

% \paragraph{Approach.}
% The intuition behind our approach is that a program establishing some aggregation property over an array necessarily has to touch all elements of the array during its execution. We can therefore interleave aggregation with the normal program execution: in the example, instead of computing \lstinline$\max$ as last step in line~18, we can compute the maximum value successively already in \lstinline$calc_pack_max$. The instrumentation makes the program \emph{easier} to verify, since it is now enough to observe that the ghost variable~\verb!ag_max! is in relationship with the program variable~\verb!bpack.max_voltage!. 

% % The approach is also easy to automate, and has been implemented in our tool~\monocera.

% As we show in \autoref{sec:instrumentation-models}, our approach gives rise to a sound method to verify programs with aggregation.
% This approach will obviously not work equally well for all programs, but it is extremely flexible and can be adapted to many operators. The discussed instrumentation is useful for programs with regular access patterns, such as the battery example, but might fail in more irregular cases. Our experimental evaluation indicates, however, that many programs are indeed regular enough to be covered, and can be verified automatically using our approach, even when considering arrays of unbounded size.

%% If the triangular numbers example should be used
% \else 

% \subsection{Motivating Example [Jesper]}

\begin{figure}[tb]
\centering
\begin{minipage}{0.87\linewidth}
\centering
  \begin{minipage}{0.44\linewidth}
\begin{lstlisting}[style=ExtWhileStyle]
// Triangular numbers
i = 0; /*A*/ s = 0; /*B*/
assume(N>0);
while(i < N) {



    i = i + 1; /*C*/

    
    s = s + i;
}


NN = N*N; /*D*/

assert(s == (NN+N)/2);
\end{lstlisting}
  \end{minipage}
  \hfill
  \begin{minipage}{0.50\linewidth}
\begin{lstlisting}[style=ExtWhileStyle]
// Instrumented program
i=0; s=0; x_sq=0; x_shad=0;
assume(N>0);
while(i < N) {
    // Begin-instrumentation
    assert(i == x_shad);
    x_sq   = x_sq + 2*i + 1;
    i      = i + 1;
    x_shad = i;
    // End-instrumentation
    s      = s + i;
}
// Begin-instrumentation
assert(N == x_shad);
NN = x_sq;
// End-instrumentation
assert(s == (NN+N)/2);
\end{lstlisting}
  \end{minipage}
\caption{Program computing triangular numbers, and its instrumented counterpart\figstop{}}
\label{fig:ex-non-linear-arithmetic-code}
\end{minipage}
\end{figure}

%The main contribution of this paper is a general framework for formulating encodings of specifications, in our context called \emph{instrumentations,} applying instrumentation automatically to programs, and systematically reasoning about the correctness of instrumentation.

%To deal with rich specifications, software model checkers 

% \begin{figure}[tb]
%     \centering
%         \lstinputlisting[style=CStyle]{examples/battery_diag_for.c}
%         \caption{Example program calculating the maximum voltage in a battery pack consisting of an array of batteries.}
%         \label{fig:battery_example}
% \end{figure}
% %\marginpar{TODO: FIG1: Synchronize syntax with defined programming in sec 2? Use normalized program? }

Consider the program on the left-hand side of \autoref{fig:ex-non-linear-arithmetic-code}, computing the \emph{triangular numbers}~$s_N = (N^2 + N) / 2$.
% We illustrate our approach on the computation of \emph{triangular numbers}~$s_N = (N^2 + N) / 2$, see left-hand side of \autoref{fig:ex-non-linear-arithmetic-code}. 
For reasons of presentation, the program has been normalised by representing the square~\lstinline!N*N! using an auxiliary variable~\lstinline!NN!. While mathematically simple, verifying the post-condition~\lstinline!s == (NN+N)/2! in the program turns out to be challenging even for state-of-the-art model checkers, since such tools are usually thrown off course by the non-linear term~\lstinline!N*N!. Computing the value of \lstinline!NN! by adding a loop in line~16 is not sufficient for most tools either, since the program in any case requires a non-linear invariant \verb!0 <= i <= N && 2*s == i*i + i! to be derived for the loop in lines~4--12.

The insight needed to elegantly verify this program is that the value~\lstinline!i*i! can be tracked during the program execution using a ghost variable~\lstinline$x_sq$. An instrumentation operator for this purpose is defined in \autoref{fig:squareInstrumentation}, instrumenting the program to maintain the relationship~\lstinline$x_sq == i*i$. The instrumented program is shown in the right-hand side of \autoref{fig:ex-non-linear-arithmetic-code}.  Initially, \lstinline$i == x_sq == 0$, and each time the value of \lstinline$i$ is modified, also the variable~\lstinline$x_sq$ is updated accordingly. In particular, we rewrite the assignments~C, D of the left-hand side program using rewrite rules~(R2) and (R4), respectively, resulting in the instrumented and correct program on the right-hand side. With the value~\lstinline$x_sq == i*i$ available, both the loop invariant and the post-condition turn into formulas over linear arithmetic, and automatic program verification becomes largely straightforward.

% \todo[inline]{JA: Add here a paragraph on instrumentation invariants and soundness}

% 

\begin{figure}[tb]
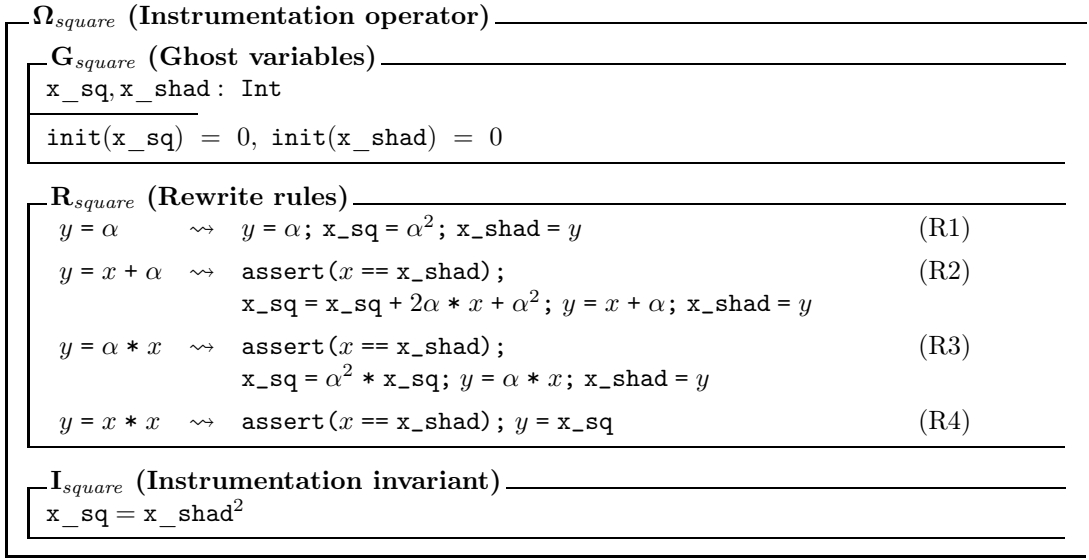

\centering
\begin{minipage}{0.87\textwidth}
    \centering
%\footnotesize
\begin{class}{\bm{\Omega}_{\mathit{square}} \textbf{
(Instrumentation operator)}}
\begin{schema}{\mathbf{G_{\mathit{square}}} \textbf{ (Ghost variables)}}
    \mathtt{x\_sq}, \mathtt{x\_shad:~Int}
    \where
    \mathtt{init(x\_sq)} ~=~ 0,\ 
    \mathtt{init(x\_shad)} ~=~ 0
\end{schema}
\\
\begin{schema}{\mathbf{R_{\mathit{square}}} \textbf{ (Rewrite~rules)}}
\begin{array}{l@{\quad}l@{\quad}l@{\qquad\qquad}l}
     y \assign \alpha  &\leadsto& y \assign \alpha\semic \verb!x_sq! \assign \alpha^2\semic \verb!x_shad! \assign y & \text{(R1)} \\[1ex]
     y \assign x \add \alpha  &\leadsto& \texttt{assert(}x \eqeq \verb!x_shad!\texttt{)}\semic  & \text{(R2)}\\
     && \verb!x_sq! \assign \verb!x_sq! \add 2\alpha \mult x \add \alpha^2\semic y \assign x \add \alpha\semic \verb!x_shad! \assign y\\[1ex]
     y \assign \alpha \mult x  &\leadsto& \texttt{assert(}x \eqeq \verb!x_shad!\texttt{)}\semic  & \text{(R3)} \\
      && \verb!x_sq! \assign \alpha^2 \mult \verb!x_sq!\semic y \assign \alpha \mult x\semic \verb!x_shad! \assign y\\[1ex]
     y \assign x\mult x  &\leadsto& \texttt{assert(}x \eqeq \verb!x_shad!\texttt{)}\semic y \assign \verb!x_sq! & \text{(R4)}
\end{array}
\end{schema}
\\
\begin{schema}{\mathbf{I_{\mathit{square}}}\textbf{ (Instrumentation invariant)}}
     \mathtt{x\_sq} = \mathtt{x\_shad}^2
\end{schema}
\end{class}
\caption{Instrumentation operator $\Omega_{square}$ for tracking squares\figstop{}}
    \label{fig:squareInstrumentation}
\end{minipage}
\end{figure}

\subsection{Automatic Program Instrumentation}
While the above examples illustrate the applicability of our instrumentation idea, we are still left with a major challenge, namely how to automatically discover the appropriate program transformation, without having to select manually the subset of statements to which the rewrite rules should be applied.
% 
% For the example, the transformed program is shown on the right-hand side of \autoref{fig:ex-non-linear-arithmetic-code}, and discussed in the next paragraphs.
% 
To this end, we split the process of program instrumentation into two parts:
\begin{enumerate}[(i)]
    \item Choosing an \emph{instrumentation operator,} which is defined manually, designed to be program-independent, and inducing a space of possible program transformations.
    \item Carrying out an automatic \emph{application strategy} to find, among the possible program transformations, one that enables verification of the program at hand.
\end{enumerate}

% An instrumentation operator for tracking squares is shown in \autoref{fig:squareInstrumentation}
% % and consists of the declaration of two ghost variables (\verb!x_sq!, \verb!x_shad!) with initial value~$0$, respectively; four rules for rewriting program statements; and the instrumentation invariant witnessing correctness of the operator.  
% The rewrite rules use formal variables~$x, y$, which can represent arbitrary variables in the program (\lstinline!i!, \lstinline!N!, \lstinline!NN!). 

% An application of the operator to a program will declare the ghost variables in the form of global variables, and then rewrite some chosen set of program statements using the provided rules. 
% Since the statements to be rewritten can be chosen arbitrarily, and since moreover multiple rewrite rules might apply to some statements, rewriting can result in many different variants of a program. 
The first step is illustrated in the examples above. The second step relies on that instrumentation operators are designed to be \emph{sound,} which means that rewriting a wrong selection of program statements might lead to an instrumented program that cannot be verified, i.e., in which assertions might fail, but instrumentation can never turn an incorrect program into a correct instrumented program. This opens up the possibility to systematically search for the right program instrumentation. We propose a counterexample-guided algorithm for this purpose, which starts from some arbitrarily chosen instrumentation, checks whether the instrumented program can be verified, and otherwise attempts to fix the instrumentation using a refinement loop. As soon as a verifiable instrumented program has been found, the search can stop and the correctness of the original program has been shown.

% To guarantee soundness of an operator, the concept of instrumentation invariants is essential. Instrumentation invariants are formulas that can (only) refer to the ghost variables introduced by an instrumentation operator, and are formulated in such a way that they hold \emph{in every reachable state of every instrumented program.} To maintain their invariants, instrumentation operators use shadow variables that duplicate the values of program variables. For example, in the operator in \autoref{fig:squareInstrumentation}, the purpose of the shadow variable~\lstinline$x_shad$ is to reproduce the value of the program variable whose square is tracked (\lstinline!i!). The rewriting rules introduce guards to detect incorrect instrumentation (the assertions in (R2), (R3), (R4)), which are particular cases in which some update of a relevant variable was missed and not correctly instrumented. The use of shadow variables and guards make instrumentation operators very flexible; in our example, note that instrumentation tracks the square of the value of \lstinline!i! during the loop, but is also used later to simplify the expression~\lstinline!N*N!. This is possible because of the instrumentation invariant and because \lstinline!i == N! holds after termination of the loop, which is verified through the assertion introduced in line~14.
% \fi

\subsection{Contributions and Outline}

%\todo[inline]{Rewrite entire subsection to fit new scope. Describe the difference with the conferences version.}

The two operators shown so far are simple and do not apply to all programs, but they can easily be generalised. The framework presented in this article provides the foundation for developing a library of formally verified instrumentation operators. In the scope of this article, we focus on two specification constructs that have been identified as particularly challenging in the literature: existential and universal \emph{quantifiers} over arrays, and \emph{aggregation} (or \emph{extended quantifiers}), which includes computing the sum or maximum value of elements in an array.
%
%We present a framework for formulating instrumentation operators similar to or generalizing the one shown in \autoref{fig:squareInstrumentation}, as well as instrumentations for a wide range of operators used in specifications. 
%
%In this paper, after formally introducing the concept of instrumentation operators (\autoref{sec:program-transformation}) and the counterexample-guided search algorithm (\autoref{sec:transformation-strategies}), we focus on two classes of instrumentation operators:
% \begin{inparaenum}[(i)]
%     \item operators for handling \emph{quantifiers} in specifications (\autoref{sec:normal-quantifiers}), which are a long-standing challenge in software model checking; and
%     \item operators for \emph{aggregation} (or \emph{extended quantifiers}), which includes computing the sum or maximum value of elements in an array (\autoref{sec:instrumentation-models}).
% \end{inparaenum}
%
Our experiments on benchmarks taken from the SV-COMP~\cite{sv-report-22} show that even relatively simple instrumentation operators can significantly extend the capabilities of a software model checker, and often make the automatic verification of otherwise hard specifications easy.

% \pr{rewrite contributions}
% \dg{Some duplication with first paragraph?}
The contributions of the article are:
\begin{inparaenum}[(i)]
  \item a general \emph{framework for program instrumentation}, which defines a space of  program transformations that work by rewriting individual statements (\autoref{sec:program-transformation} and \autoref{sec:operator-composition});
  \item a library of \emph{concrete operators} defined within the framework, obtaining full coverage of the quantifiers (\autoref{sec:normal-quantifiers}) and extended quantifiers (\autoref{sec:instrumentation-models}) in ACSL~\cite{acsl};
  \item machine-checked proofs of the correctness of the instrumentation operators for quantifiers
  %quantifiers~$\forall$ 
  and 
  %the extended quantifier~\verb!\max!; 
  extended quantifiers; 
  \item correctness results showing \emph{soundness} and \emph{weak completeness} for all instrumentation operators, and \emph{completeness} for certain classes of programs (\autoref{sec:instrumentation-correctness});
  \item an application strategy \emph{search algorithm} in this space, for a given program (\autoref{sec:transformation-strategies});
  \item a new \emph{verification tool}, \monocera, that is tailored to the verification of programs with aggregation; and
  \item an \emph{evaluation} of our method and tool on a set of examples, including such from SV-COMP~\cite{sv-report-22} (\autoref{sec:evaluation}).
\end{inparaenum}
\section{Instrumentation Framework}
\label{sec:program-transformation}

% \marginpar{[3 pages]}
% \marginpar{[Dilian, Christian]}

% \dg{While the paper addresses how to verify programs with aggregation, our approach is more general. Essentially, we transform a program into one operating over an extended state space, where the verification problem at hand is easier to reason about than in the original program. The transformation is generally semantics-preserving over the original program domain.}

%In this section, we formally introduce the instrumentation framework. The framework is instantiated in \autoref{sec:instrumentation_operators_for_arrays} to obtain operators for quantification and aggregation over arrays.
%In \autoref{sec:instrumentation-correctness} we discuss correctness properties of instrumentation, and in \autoref{sec:transformation-strategies} show an algorithm for automatic instrumentation.
% To describe and understand our program instrumentations more
%systematically, 

We now formally introduce the instrumentation framework.
The instrumentation of a program requires two main ingredients:
\begin{enumerate}
\item
An \emph{instrumentation operator} that defines how to
rewrite program statements
% \lstinline{select}, \lstinline{store} and \lstinline{aggregate}
with the purpose of eliminating language constructs that are
difficult to reason about automatically,
but leaves the choice of which occurrences of these statements to
rewrite to the second part (this section).
% transforms programs~$P$ to instrumented programs~$\Omega(P)$. The operator is highly non-deterministic, since there are many ways to instrument.
%This section is dedicated to this first part.
\item
An \emph{application strategy} for the instrumentation operator, which can be implemented using heuristics or systematic search, among others. The strategy is responsible for selecting the right (if any) program instrumentation from the many possible ones. 
% , i.e., an instrumentation that enables verification of programs with aggregation.
% There might be various approaches that could be used here; for instance, systematic enumeration; greedy methods; counterexample-guided methods, etc.
\autoref{sec:transformation-strategies} is dedicated to the second part.
\end{enumerate}
Even though instrumentation operators are non-deterministic, we shall guarantee their \emph{soundness:} if the original program has a failing assertion, so will any instrumented program, regardless of the chosen application strategy; that is, instrumentation of an incorrect program will never yield a correct program.
This aspect is discussed in \autoref{sec:instrumentation-correctness}.

% Under certain additional assumptions, 
We shall also guarantee a
weak form of \emph{completeness}, to the effect that
if an assertion that has not been added to the program by the instrumentation fails in the instrumented program, then it will also fail in the original program. As a result, any counterexample (for such an assertion) produced when verifying the instrumented program can be transformed into a counterexample for the original program.

Finally, for certain classes of programs, in \autoref{subsec:instrumentation-completeness} we shall guarantee \emph{completeness} in the typical, stricter sense. That is, if a program is correct and adheres to certain conditions, then there is an instrumented program that is also correct, such that automatic verification has been made easier.

%%%%%%%%%%%%%%%%%%%%%%%%%%%%%%%%
%\newpage

\begin{figure}[tb]
    \begin{align*}
        \nonTerm{Type} ~::=~~ & \texttt{Int} \mid \texttt{Bool} \mid \texttt{Array}~\nonTerm{Type}
        \\
        \nonTerm{Expr} ~::=~~ & \nonTerm{DecimalNumber} \mid \texttt{true} \mid \texttt{false} \mid \nonTerm{Variable}
        \\ \mid ~~~ &
        \nonTerm{Expr} \eqeq \nonTerm{Expr} \mid
        \nonTerm{Expr} \;\texttt{<=}\; \nonTerm{Expr} \mid
        \texttt{!}\nonTerm{Expr} \mid \nonTerm{Expr} \;\texttt{\&\&}\; \nonTerm{Expr} 
        \\ \mid ~~~ &
        \nonTerm{Expr} \;\texttt{||}\; \nonTerm{Expr} \mid
        \nonTerm{Expr} \add \nonTerm{Expr} \mid \nonTerm{Expr} \mult \nonTerm{Expr} \mid \nonTerm{Expr} / \nonTerm{Expr}
        \\ \mid ~~~ &
    %    \\ \mid ~~~ &
       \texttt{const(}\nonTerm{Expr}\texttt{)} \!\mid\!
        \texttt{select(}\nonTerm{Expr}\texttt{,} \nonTerm{Expr}\texttt{)}
        \!\mid\!
        \texttt{store(}\nonTerm{Expr}\texttt{,} \nonTerm{Expr}\texttt{,} \nonTerm{Expr}\texttt{)} 
        %\mid 
        %\texttt{aggregate}_{M, h}\texttt{(}\nonTerm{Expr}\texttt{,} \nonTerm{Expr}\texttt{,} \nonTerm{Expr}\texttt{)}
        \\
        \nonTerm{Prog} ~::=~~ & 
        \texttt{skip} \mid
        \nonTerm{Variable} \assign \nonTerm{Expr}
        \mid
        \nonTerm{Prog}\semic \nonTerm{Prog} \mid
        \texttt{while}~\texttt{(}\nonTerm{Expr}\texttt{)}~\nonTerm{Prog}
        \\ \mid ~~~ &
        \texttt{assert(}\nonTerm{Expr}\texttt{)} \mid
        \texttt{assume(}\nonTerm{Expr}\texttt{)} \mid
        \texttt{if}~\texttt{(}\nonTerm{Expr}\texttt{)}~\nonTerm{Prog}~\texttt{else}~\nonTerm{Prog}
    \end{align*}
    \caption{Syntax of the core language\figstop{}}
    \label{tab:programs}
\end{figure}

%%%%%%%%%%%%%%%%%%%%%%%%%%%%%%%%%%%%%%%%%%%%%%%%%%%%%%%%%%%%%%%%

%%%%%%%%%%%%%%%%%%%%%%%%%%%%%%%%%%%%%%%%%%%%%%%%%%%%%%%%%%%%%%%%
\subsection{The Core Language}
\label{subsec:programming-language}

A simple programming language, containing loops and arrays, is used to present the instrumentation framework.

\paragraph{Syntax}

While
% \todo[]{
% %Add division to core language (syntax and type rules). Any other things?
% Add constant array to syntax, change initialisation of arrays in instr. op. below, to either constant or nondet.}
%
our implementation works on programs represented as constrained Horn clauses~\cite{DBLP:conf/birthday/BjornerGMR15}, i.e., is language-agnostic,
for readability purposes we present our approach in the setting of an imperative core programming language with data-types for unbounded integers, Booleans, and arrays, and \texttt{assert} and \texttt{assume} statements. The language is deliberately kept simple, but is still close to
standard C. The main exception is the semantics of arrays: they are defined here to be 
% MV: Added sentence on array equality as this was not-obvious to me when starting out.
\emph{functional} and therefore represent a value type. A notable deviation from C-like, imperative-style arrays is that the equality of two arrays depends on the array contents, and not on a pointer value. 
Arrays have integers as index type and are unbounded, and their signature and semantics are otherwise borrowed from the SMT-LIB theory of extensional arrays~\cite{smtlib}:
\begin{itemize}
\item
\makebox[0.68\linewidth][l]{%
 Constructing a \emph{constant} array filled with \texttt{x}:} \texttt{const(x)};
\item
\makebox[0.68\linewidth][l]{%
  \emph{Reading} the value of an array~\texttt{a} at index \texttt{i}:} \texttt{select(a, i)};
\item
\makebox[0.68\linewidth][l]{%
\emph{Updating} an array~\texttt{a} at index~\texttt{i} with 
a new value~\texttt{x}:} \texttt{store(a, i, x)};
\item
\makebox[0.68\linewidth][l]{%
\emph{Comparing} array~\texttt{a} and~\texttt{b} for equality:} \texttt{a == b}.
\end{itemize}

\begin{figure}[tb]
    \begin{gather*}
      \infer{\typeJudg{\nonTerm{DecimalNumber}}{\texttt{Int}}}{}
      \qquad
      \infer{\typeJudg{\texttt{true}}{\texttt{Bool}}}{}
      \qquad
      \infer{\typeJudg{\texttt{false}}{\texttt{Bool}}}{}
      \qquad
      \infer{\vphantom{X}\typeJudg{x}{\sigma}}{x \in \mathcal{X}, \alpha(x) = \sigma}
      \\[1ex]
      \infer{\typeJudg{s \eqeq t}{\texttt{Bool}}}{\typeJudg{s}{\sigma} & \typeJudg{t}{\sigma}}
      \qquad
      \infer{\typeJudg{s  \;\texttt{<=}\; t}{\texttt{Bool}}}{\typeJudg{s}{\texttt{Int}} & \typeJudg{t}{\texttt{Int}}}
      \qquad
      \infer{\typeJudg{s  \;\texttt{+}\; t}{\texttt{Int}}}{\typeJudg{s}{\texttt{Int}} & \typeJudg{t}{\texttt{Int}}}
      \qquad
      \infer{\typeJudg{s  \;\texttt{*}\; t}{\texttt{Int}}}{\typeJudg{s}{\texttt{Int}} & \typeJudg{t}{\texttt{Int}}}
      \\[1ex]
      \infer{\typeJudg{s  \;\texttt{/}\; t}{\texttt{Int}}}{\typeJudg{s}{\texttt{Int}} & \typeJudg{t}{\texttt{Int}}}
      \qquad
     \infer{\typeJudg{ \texttt{!} c}{\texttt{Bool}}}{\typeJudg{c}{\texttt{Bool}}}
      \qquad
      \infer{\typeJudg{s \;\texttt{\&\&}\; t}{\texttt{Bool}}}{\typeJudg{s}{\texttt{Bool}} & \typeJudg{t}{\texttt{Bool}}}
      \qquad
      \infer{\typeJudg{s \;\texttt{||}\; t}{\texttt{Bool}}}{\typeJudg{s}{\texttt{Bool}} & \typeJudg{t}{\texttt{Bool}}}
      \\[1ex]
     \infer{\typeJudg{\texttt{const(}t\texttt{)}}{\texttt{Array}~\sigma}}{\typeJudg{t}{\sigma}}
     \qquad
      \infer{\typeJudg{\texttt{select(}a\texttt{,} t\texttt{)}}{\sigma}}{\typeJudg{a}{\texttt{Array}~\sigma} & \typeJudg{t}{\texttt{Int}}}
      \qquad
      \infer{\typeJudg{\texttt{store(}a\texttt{,} s\texttt{,} t\texttt{)}}{\texttt{Array}~\sigma}}{\typeJudg{a}{\texttt{Array}~\sigma} & \typeJudg{s}{\texttt{Int}} & \typeJudg{t}{\sigma}}
%      \\[1ex]
%      \infer{\typeJudg{\texttt{aggregate}_{M, h}\texttt{(}a\texttt{,} l\texttt{,} u\texttt{)}}{\sigma_M}}{\typeJudg{a}{\texttt{Array}~\sigma} & \typeJudg{l}{\texttt{Int}} & \typeJudg{u}{\texttt{Int}} & h : D_{\sigma}^* \to D_{\sigma_M}}
    \end{gather*}
    \caption{Typing rules of the core language\figstop{}}
    \label{tab:typing}
\end{figure}

The complete syntax of the core language is given in \autoref{tab:programs}.
Programs are written using a vocabulary~$\mathcal{X}$ of typed program variables;
the typing rules of the language are shown in \autoref{tab:typing}.
%
%the typing rules of the language are given in
%\iftr Appendix~\ref{app:typing-rules}\else \cite{techreport}\fi.
As syntactic sugar, we sometimes write \verb!a[i]! instead of \texttt{select(a, i)}, and 
\verb!a[i] = x! instead of \texttt{a = store(a, i, x)}. %, and use the derived operators~\lstinline!&&! and \lstinline!||!.
We also use \lstinline$a = _$ to mean that \lstinline!a! is assigned to some constant array, for which we do not care about the value of the elements. 

\paragraph{Semantics}

We assume the Flanagan-Saxe \emph{extended execution model} of programs with 
\texttt{assume} and 
\texttt{assert} statements (see, e.g., \cite{fla-sax-01-popl}),  in which
executing 
an \texttt{assert} 
statement with an argument that evaluates to \textsf{false} 
\emph{fails}, i.e., terminates abnormally. An \texttt{assume}
statement with an argument that evaluates to \textsf{false} has the same
semantics as a non-terminating loop.
We denote by $D_{\sigma}$ the domain of a program type~$\sigma$. The domain
of an array type~$\mathtt{Array}~\sigma$ is the set of
functions~$f : \mathbbm{Z} \to D_{\sigma}$.
Partial correctness properties of programs are expressed using \emph{Hoare triples} $\{\mathit{Pre}\} \;P\; \{\mathit{Post}\}$, which state that an execution of $P$, starting in a state satisfying $\mathit{Pre}$, never fails, and may only terminate in states that satisfy $\mathit{Post}$.
As usual, a program~$P$ is considered \emph{(partially) correct} if the Hoare triple~$\{\mathit{true}\} \;P\; \{\mathit{true}\}$ holds.

%To reason formally about program executions, we consider
%a structural (or small-step) operational semantics (SOS)
%of the language, built over \emph{configurations} 
%$\econfig{p}{s}{q}$, where $p$ is a \emph{control point} 
%(or \emph{program counter}), $s$
%a state that maps program variables to values from their respective domains~$D_\sigma$,
%and $q \in \set{\top, \bot}$ a flag signifying whether the configuration is normal 
%or abnormal.

The evaluation of program expressions is modelled
using a function~$\llbracket \cdot \rrbracket_s$ that maps program
expressions~$t$ of type~$\sigma$ to their value~$\llbracket t \rrbracket_s \in D_{\sigma}$
in the state~$s$.
%\todo[inline]{Extend language with procedure calls?}
%
%The rules of the SOS for normal execution are standard and 
%are omitted here; the only situation that results
%in an abnormal configuration is when executing an 
%\texttt{assert} statement with an expression that evaluates 
%to \textsf{false} in the current state~$s$.
%

%%%%%%%%%%%%%%%%%%%%%%%%%%%%%%%%%%%%%%%

\subsection{Instrumentation Operators}
\label{sec:instrumentation_operators}
% We start with an abstract definition of instrumentation operators.
An instrumentation operator defines schemes to rewrite programs
while preserving the meaning of the existing program assertions. Without
loss of generality, we restrict program rewriting to assignment
statements. Instrumentation can introduce
\emph{ghost state} by adding arbitrary fresh variables to the program.
%
%,
%with the purpose of replacing language constructs that are
%hard to reason about mechanically with constructs that are 
%easier to handle.
The main part of an instrumentation consists of \emph{rewrite rules},
which are schematic rules $r \;\texttt{=}\; t \leadsto s$, where
the meta-variable~$r$ ranges over program variables, $t$~is an expression
that can contain further meta-variables (but not the left-hand side variable~$r$),
and $s$ is a schematic program in which
the meta-variables from $r \;\texttt{=}\; t$ might occur. Any assignment that matches $r \;\texttt{=}\; t$
can be rewritten to~$s$.

% \dg{We should explain the idea behind the definition.
% I believe that it is as follows: 
% %
% Under an invariant condition~$I$ on the ghost state, called the 
% \emph{instrumentation invariant}, the instrumentation is
% \emph{semantics-preserving} w.r.t.\ the original program, 
% so that the execution of all \texttt{assert} 
% statements in the original program before and after 
% instrumentation has the same effect.}

\begin{definition}[Instrumentation Operator]
\label{def:instr-op}
An \emph{instrumentation operator}
% for an aggregation function~$\texttt{aggregate}_{M, h}$ defined by the monoid~$M$ and the homomorphism~$h$
is a tuple
%
% $\Omega = (G; i_\mathit{sel}, i_\mathit{str}, i_\mathit{agg}; \pi)$, where:
\instropdef{}, where:
\begin{itemize}
    \item[($i$)]
    % $G = \langle \sigma_1, \ldots, \sigma_k \rangle$
    $G = \langle (\mathtt{x}_1, \mathit{init}_1), \ldots, (\mathtt{x}_k, \mathit{init}_k) \rangle$
    is a tuple of pairs of ghost variables and their initial values;
    % \marginpar{\pr{``$G = \langle \sigma_1, \ldots, \sigma_k \rangle$ is a tuple of program types of ...''?}}
    % \item[($ii$)]
    % $S$ is a tuple of the initial values of the ghost variables.
    \item[($ii$)]
    % $i_\mathit{sel}, i_\mathit{str}, i_\mathit{agg}$
    $R$
    % is the \emph{instrumentation code} for select, store, and aggregate
    % expressions, respectively;
    % \dg{Present as (context-free) re-write rules?}
    is a set of rewrite rules
    $r \;\texttt{=}\; t \leadsto s$, where $s$ is a program operating on the ghost
    variables~$\mathtt{x}_1, \ldots, \mathtt{x}_k$ (and containing meta-variables from $r \;\texttt{=}\; t$);
    % $r = e \leadsto s$,
    % $s$ is one of the normalised select, store, and aggregate statements.
%    where $e$ %\lstinline{e}
%    is an expression assigned to the variable
%    $r$ %\lstinline{r},
%    and $s$ is a program (possibly several statements).
    \item[($iii$)]
    $I$ is a formula over the ghost variables~$\mathtt{x}_1, \ldots, \mathtt{x}_k$, called the
    \emph{instrumentation invariant.}
    % \item[($v$)]
    % $\pi : D_{G} \to M$ is a projection of a valuation of the ghost variables to the monoid~$M$, with the domain descriptor~$D$ lifted to tuples.
    
    % \pr{This should be the simplest case; it could also give bounds.}
    \end{itemize}
The rewrite rules~$R$ and the invariant~$I$ must adhere to the following constraints:
\begin{enumerate}
    \item The instrumentation invariant $I$ is satisfied by the initial ghost values, i.e., it holds in the state
     $\{\mathtt{x}_1 \mapsto \mathit{init}_1, \ldots, \mathtt{x}_k \mapsto \mathit{init}_k\}$.
    \item For all rewrites
        $r \assign t \leadsto s \in R$
        the following conditions hold:
        \begin{enumerate}
            \item
                $s$ terminates (normally or abnormally) for 
                pre-states satisfying $I$, assuming that all meta-variables
                are ordinary program variables.
            \item $s$ does not assign to variables other than
                $r$ or the ghost variables~$\mathtt{x}_1, \ldots, \mathtt{x}_k$.
            \item $s$ preserves the instrumentation invariant:
                $\{ I \}\ s'\ \{ I \}$, where $s'$ is $s$
                with every $\texttt{assert(}e\texttt{)}$ statement replaced by an $\texttt{assume(}e\texttt{)}$ statement.
            \item $s$ preserves the semantics of the assignment~$r \assign t$, in the sense that the Hoare triple
                $\{ I \} \;\texttt{z} \assign t\semic s' \; \{ \texttt{z} = r \}$,
                where $\texttt{z}$ is a fresh variable, holds.
            % \item $s$ preserves the values of the program variables.
        \end{enumerate}
\end{enumerate}
\end{definition}

% The rationale behind this definition is 
% as follows.
% that under 
% an invariant condition~$I$ on the ghost state, called
The conditions imposed in the definition ensure
that all instrumentations are \emph{correct}, in the sense
that they are sound and weakly complete, as we show below. 
In particular, the instrumentation invariant guarantees that the rewrites of program
statements are \emph{semantics-preserving} w.r.t. the original
program, and thus, the execution of any \lstinline{assert}
statement of the original program has the same effect
before and after instrumentation.
%
% In particular,
% the correctness of the instrumented program in the example implies the correctness
% of the original program.
The conditions can themselves be deductively verified to hold
for each concrete instrumentation operator, and that this check is \emph{independent} of the programs to be instrumented, so that
an instrumentation operator can be proven correct once and for all.

Importantly, an instrumentation operator~$\Omega$ does itself not define 
which occurrences of program statements are to be rewritten, 
but only how they are rewritten.
\begin{definition}[Rewriting a program statement]
  \label{def:matchingRule}
  A program statement~\lstinline{v = e} matches a rewrite rule 
  $r \assign t \leadsto s$ if there is a substitution~$\sigma$, replacing each meta-variable with a program expression, that maps $r$ to $\sigma(r) = \texttt{v}$ and
  $t$ to $\sigma(t) = \texttt{e}$.

  The result of rewriting \lstinline{v = e} using the rule~$r \assign t \leadsto s$ is defined as follows:
  \begin{enumerate}[(i)]
  \item If \lstinline$v$ does not occur in \lstinline$e$, the result of rewriting~\lstinline{v = e} is the program~$\sigma(s)$, i.e., the right-hand side~$s$ of the rule, with meta-variables substituted with the expressions occurring in the assignment~\lstinline{v = e}.
  \item If \lstinline$v$ occurs in \lstinline$e$, we first replace \lstinline{v = e} with the two statements \lstinline!v' = e; v = v'!, for some fresh variable~\lstinline!v'! that has the same type as \lstinline!v!, and then rewrite the new assignment~\lstinline!v' = e! using $r \assign t \leadsto s$.
  \end{enumerate}
\end{definition}

The procedure in second case~(ii) has the purpose of ensuring that the right-hand side~\lstinline!e! is stable in $s$; otherwise, carrying out the assignment could have the effect of changing the value of expressions in \lstinline!e!.

Given a program~$P$ and the operator~$\Omega$, an instrumented program~$P'$
is derived by carrying out the following two steps:
\begin{inparaenum}[(i)]
\item variables $\mathtt{x}_1, \ldots, \mathtt{x}_k$ and the 
 assignments $\mathtt{x}_1 \assign \mathit{init}_1\semic \ldots\semic \mathtt{x}_k \assign \mathit{init}_k$
  are added at the beginning of the program, and
\item some of the assignments in~$P$ that match rewrite rules~$r \assign t \leadsto s$ in $\Omega$ are rewritten as defined in \autoref{def:matchingRule}.
\end{inparaenum}
We denote by~$\Omega(P)$ the set of all instrumented programs~$P'$ that can be 
derived in this way.
%
\iffalse
An example of an instrumentation operator and its
application was shown \autoref{fig:ex-non-linear-arithmetic-code}
and \autoref{fig:squareInstrumentation}.
\fi

% \marginpar{CL: will we have example of operator in intro? No!}

Two simplified examples of instrumentation operators were already given in \autoref{sec:introduction}. The reader is advised to return to those examples, and convince herself that the operators indeed satisfy the conditions required in \autoref{def:instr-op}.

\section{Instrumentation Operators for Arrays}
\label{sec:instrumentation_operators_for_arrays}

We will now instantiate our framework for several classes of functions over arrays, targeting, in particular, the different kinds of quantifiers over arrays provided by ACSL~\cite{acsl}. Among others, we will obtain a more general and more precise version of the operator for the extended quantifier~\lstinline$\max$ from \autoref{sec:introduction}.

%%%%%%%%%%%%%%%%%%%%%%%%%%%%%%%%%%%%%%%%%%%%%%%%%%%%%%%%%%%%%%%%
%%%%%%%%%%%%%%%%%%%%%%%%%%%%%%%%%%%%%%%%%%%%%%%%%%%%%%%%%%%%%%%%

\subsection{Instrumentation Operators for Quantification over Arrays}
\label{sec:normal-quantifiers}

To handle quantifiers in a programming setting, we extend the language defined in \autoref{tab:programs} by adding quantifier expressions over arrays, as shown in \autoref{tab:programs_quantifiers}, where $Q$ ranges over $\{\forall, \exists\}$.
% As shorthand, we
We shall often use $\texttt{\textbackslash forall}$ for $\texttt{quant}_\forall$, and $\texttt{\textbackslash exists}$ for $\texttt{quant}_\exists$. 
% As seen, we 
We also extend the language with a lambda expression over two variables. The rationale is that quantified properties can often be expressed as a binary predicate, with the first argument corresponding to the value of an element and the second to the index. 

\begin{figure}[tb]
    \begin{align*}
        \nonTerm{Expr} ~::=~~ &\texttt{(}\lambda \texttt{(}\nonTerm{Variable}\texttt{,}\nonTerm{Variable}\texttt{)}\texttt{.}\nonTerm{Expr}\texttt{)}~\texttt{(}\nonTerm{Expr},~\nonTerm{Expr}\texttt{)} \mid
        \\
        &\texttt{quant}_{Q}\texttt{(}\nonTerm{Expr}\texttt{,} \nonTerm{Expr}\texttt{,} \nonTerm{Expr}\texttt{,} \lambda\texttt{(}\nonTerm{Variable},\nonTerm{Variable}\texttt{)}.\nonTerm{Expr}\texttt{)} \mid
        % \\
        % &\texttt{exists}\texttt{(}\nonTerm{Expr}\texttt{,} \nonTerm{Expr}\texttt{,} \nonTerm{Expr}\texttt{,} \lambda \texttt{(}\nonTerm{Variable},\nonTerm{Variable}\texttt{)}.\nonTerm{Expr}\texttt{)}
    \end{align*}
    \caption{Extension of the core language with quantified expressions\figstop{}}
    \label{tab:programs_quantifiers}
\end{figure}

% This allows us to express properties over both the value of an element and its index. 
% We formalise \texttt{quant} in terms of monoid homomorphism. For reasons of presentation, the formalisation is deferred to \autoref{sec:instrumentation-models}, while we proceed here by illustrating the usage of \texttt{forall}.

% , we can express that each element should be
% \pr{in the example it is actually ``equal''?}
% greater than or
% equal to its index. 

% as is done in the example program in \autoref{fig:quantified_example}. In the program, each element in the array is assigned the value corresponding to its index, after which it is asserted that this property indeed holds. 

% \pr{The dot is used inconsistently in $\lambda$ terms. How about defining a macro for lambda terms, including also the right spacing?}
Using $\texttt{P(x$_0$,i$_0$)}$ as shorthand for $\texttt{(}\mathtt{\lambda}\texttt{(x,i).P)(x$_0$,i$_0$)}$, the semantics of the new expressions can be defined formally as:
\begin{align*}
    % \llbracket (\mathtt{\lambda(x,i).P)~(e_1,~e_2)}  \rrbracket_s &~=~ \llbracket \mathtt{P[x\mapsto e_1, i \mapsto e_2]} \rrbracket_s
    % \\
    \llbracket \texttt{quant}_Q\texttt{(a, l, u, }\lambda \texttt{(x,i).P}\texttt{)} \rrbracket_s
    &~=~
    Q \mathtt{i \in [l, u).}\; \llbracket \texttt{P(a[i],i)} \rrbracket_s
    % \\
    % \llbracket \texttt{exists}\texttt{(a, l, u, }\lambda \texttt{(x,i).P}\texttt{)} \rrbracket_s
    % &~=~
    % \exists \mathtt{i \in [l, u).}\; \llbracket \texttt{P(a[i],i)} \rrbracket_s
\end{align*}
% % \dg{Why is $P(a[i])$ not enough?}
% Note that the types of \texttt{x} and \texttt{a} must be compatible and \texttt{P} be a Boolean-valued expression. 

\begin{figure}[tb]
    \centering
    \begin{minipage}{0.87\textwidth}
    \centering
\begin{lstlisting}[escapechar=@,style=ExtWhileStyle]
Int N = nondet;
assume(N > 0);
Array Int a = const(0);
Int i = 0;
while(i < N) {
    a = store(a, i, i);
    i = i + 1;
}
Bool b = \forall(a, 0, N, @$\lambda$@(i,x).(x == i));
assert(b);
\end{lstlisting}
    \caption{Example of program to be verified using a quantified assert statement\figstop{}}
    \label{fig:quantified_example}
    \end{minipage}
\end{figure}

%exists$_{\geq N-1}$(a, 0, N)

\begin{figure}[!htb]
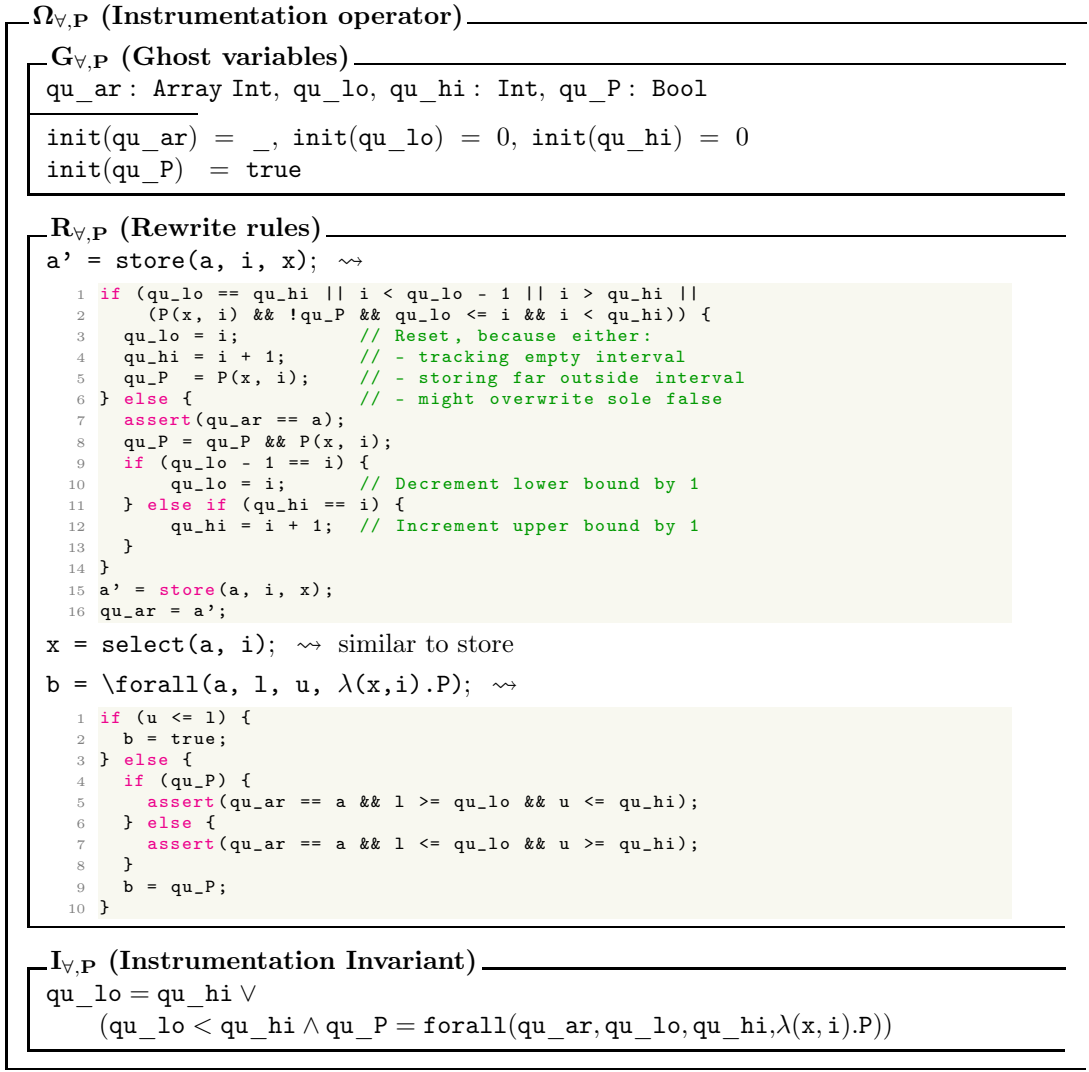

    \centering
    \begin{minipage}{0.87\textwidth}
    \centering
% \footnotesize
\begin{class}{\bm{\Omega}_{\mathbf{\forall,P}} \textbf{
(Instrumentation operator)}}
\begin{schema}{\mathbf{G_{\forall,P}} \textbf{ (Ghost variables)}}
    \mathtt{qu\_ar:~Array~Int},\
    \mathtt{qu\_lo,~qu\_hi:~Int},\ 
    \mathtt{qu\_P:~Bool}
    \where
    \mathtt{init(qu\_ar)} ~=~ \_,\ 
    \mathtt{init(qu\_lo)} ~=~ 0,\ 
    \mathtt{init(qu\_hi)} ~=~ 0 \\
    \mathtt{init(qu\_P)} ~~=~ \mathtt{true}
\end{schema}\\
\begin{schema}{\mathbf{R_{\forall,P}} \textbf{ (Rewrite~rules)}}
     \texttt{a' = store(a, i, x)}; ~\leadsto~  \\
     \qquad \mbox{
\begin{lstlisting}[linewidth=\opcodesnippetwidth, style=ExtWhileStyleOp]
if (qu_lo == qu_hi || i < qu_lo - 1 || i > qu_hi || 
    (P(x, i) && !qu_P && qu_lo <= i && i < qu_hi)) {
  qu_lo = i;          // Reset, because either:
  qu_hi = i + 1;      // - tracking empty interval
  qu_P  = P(x, i);    // - storing far outside interval
} else {              // - might overwrite sole false
  assert(qu_ar == a);
  qu_P = qu_P && P(x, i);
  if (qu_lo - 1 == i) {
      qu_lo = i;      // Decrement lower bound by 1
  } else if (qu_hi == i) {
      qu_hi = i + 1;  // Increment upper bound by 1
  }
}
a' = store(a, i, x);
qu_ar = a';
\end{lstlisting}}
     \also
     \texttt{x = select(a, i)}; ~\leadsto~  \text{similar to store} 
     \also
     \texttt{b = \textbackslash forall(a, l, u, } \lambda \texttt{(x,i).P)}; ~\leadsto~ \\
     \qquad \mbox{
\begin{lstlisting}[linewidth=\opcodesnippetwidth, style=ExtWhileStyleOp]
if (u <= l) {
  b = true;
} else {
  if (qu_P) {
    assert(qu_ar == a && l >= qu_lo && u <= qu_hi);
  } else {
    assert(qu_ar == a && l <= qu_lo && u >= qu_hi);
  }
  b = qu_P;
}
\end{lstlisting}}
\end{schema}\\
\begin{schema}{\mathbf{I_{\forall,P}}\textbf{ (Instrumentation Invariant)}}
     \mathtt{qu\_lo} = \mathtt{qu\_hi} \: \lor \\ \qquad
     (\mathtt{qu\_lo} < \mathtt{qu\_hi}\land %\\\hspace{79pt} 
     \mathtt{qu\_P} = 
     \mathtt{forall(qu\_ar, qu\_lo, qu\_hi,} \lambda\mathtt{(x, i).P))}
\end{schema}
\end{class}
\caption{Instrumentation operator for universal quantification\figstop{}}
    \label{fig:instr-op-forall}
    \end{minipage}
\end{figure}

As an example, consider the program in \autoref{fig:quantified_example}, which uses $\texttt{forall}$ to specify that every element in the array should equal its index. To verify this program, we
define the instrumentation operator $\forallinstrop{}$ in \autoref{fig:instr-op-forall}. The operator is similar to $\maxinstrop{}$, which was introduced in \autoref{sec:motivating_examples}, but, instead of tracking the maximum value, we use the ghost variable \lstinline$qu_P$ to track if the given property $P$ holds for all elements in the tracked interval [\lstinline$qu_lo$, \lstinline$qu_hi$). Naturally, an instrumentation operator for existential quantification can be defined in a similar fashion. 
When applying instrumentation operators for array properties, such as $\forallinstrop{}$,  we assume for simplicity a \emph{normal form} of programs, into which every program can be rewritten by introducing additional variables. In the normal form, \lstinline{store}, \lstinline{select} and \lstinline{forall} can only occur in simple assignment statements. For example, \lstinline!store! is restricted to occur in statements of the form: \lstinline{a' = store(a, i, x)}.
%
\iffalse
\todo[inline]{PR: there is a bug in the operator for $\forall$. We do not require that the different meta-variables in a' = store(a, i, x) represent distinct variables, in fact meta-variables can represent arbitrary (side effect-free) program expressions. According to our definition, we could rewrite also the statement a = store(a, i, select(a, i)+1). It is therefore important that the statement that is rewritten is executed only /at the end/ of the ghost code snippets. Executing it at the beginning could change the values of i, x, and the semantics of the assert statements in the ghost code becomes incorrect.

The same bug occurs in two of the operators in the CAV paper.}
\fi

The rewrite rules can be justified as follows.
For \lstinline!store!, the first if-branch resets the tracking to the singleton interval $[\mathtt{i}, \mathtt{i+1})$ when accessing elements far outside of the currently tracked interval, or if we are tracking the empty interval (as is the case at initialisation).
% and, if so, we reset the tracked interval to $[\mathtt{i}, \mathtt{i+1})$.
If an access occurs immediately adjacent
to the currently tracked interval (e.g., if $\mathtt{i} = \mathtt{qu\_lo-1}$), then that element is added to the tracked interval, and the value of $\mathtt{qu\_P}$ is updated to also account for the value of $\mathtt{P}$ at index $\mathtt{i}$.
If instead the access is within the tracked interval, then we either reset the interval (if $\mathtt{qu\_P}$ is $\mathtt{false}$) or keep the interval unchanged (if $\mathtt{qu\_P}$ is $\mathtt{true}$).
Rewrites of \lstinline!select! are similar to \lstinline!store!, except tracking does not
need to be reset when reading inside the tracked interval.
For rewrites of quantified expressions, if the quantified interval is empty, $\mathtt{b}$ is assigned $\mathtt{true}$.
Otherwise, assertions check that the tracked interval matches the quantified interval before assigning $\mathtt{t}$ to $\mathtt{qu\_P}$.
If $\mathtt{qu\_P}$ is $\mathtt{true}$, then it is sufficient that quantification occurs over a sub-interval of the tracked interval, and vice versa if $\mathtt{qu\_P}$ is $\mathtt{false}$.

%%%%%%%%%%%%%%
\begin{figure}[!tb]
    \centering
    \begin{minipage}{0.87\textwidth}
\begin{lstlisting}[escapechar=@,style=ExtWhileStyle]
Int qu_lo = 0; qu_hi = 0; Int qu_ar = []; Bool qu_P = true;
Int N = nondet; assume(N > 0); Array Int a = const(0); Int i = 0;
while(i < N) {
    // Begin-instrumentation
    if (qu_lo == qu_hi || i < qu_lo - 1 || i > qu_hi || 
        (i == i && !qu_P && qu_lo <= i && i < qu_hi)) {
      qu_lo = i; qu_hi = i + 1; qu_P = (i == i);
    } else {
      assert(qu_ar == a); qu_P = qu_P && i == i;
      if (qu_lo - 1 == i) {
          qu_lo = i;
      } else if (qu_hi == i) {
          qu_hi = i + 1;
      }
    }
    a' = store(a, i, i);
    qu_ar = a';
    // End-instrumentation
    a = a'; i = i + 1;
}
Bool b;
// Begin-instrumentation
if (N <= 0) {
  b = true;
} else {
  if (qu_P) {
    assert(qu_ar == a && 0 >= qu_lo && N <= qu_hi);
  } else {
    assert(qu_ar == a && 0 <= qu_lo && N >= qu_hi);
  }
  b = qu_P;
}
// End-instrumentation
assert(b);
\end{lstlisting}
    \caption{Resulting program from applying the instrumentation $\Omega_{\forall, \lambda(i,x). x = i}$ to the program in \autoref{fig:quantified_example}\figstop{}}
    \label{fig:quantified_example_instr}
    \end{minipage}
\end{figure}
%%%%%%%%%

% \pr{this paragraph should be made more positive!}
The result of applying $\Omega_{\forall, \lambda(i,x). x = i}$ to the program in \autoref{fig:quantified_example} is shown in \autoref{fig:quantified_example_instr}. In the process of
rewriting, we renamed the left-hand side of \lstinline!a = store(a, i, i)! by introducing a fresh variable~\lstinline!a'!, following case~(ii) of \autoref{def:matchingRule}.
%\iftr Appendix~\ref{app:instrexample_quantifier}\else \cite{techreport}\fi.
As exhibited by the experiments in \autoref{sec:evaluation}, the resulting program is, in many cases, easier to handle for state-of-the-art verification tools.
%
% We also stress again that $\Omega_{\forall, P}$, as defined above,
% is only one among infinitely many possible instrumentation operators
% for universally quantified expressions.
Note that the proposed instrumentation operator is only one possibility among many. For example, one could track simultaneously several ranges over the array in question, or also track the index of some element in the array over which \texttt{P} holds, or make different choices on store operations outside of the tracked interval.

\subsection{Instrumentation Operators for Aggregations}\label{sec:instrumentation-models}
We now turn to defining instrumentation operators for any \emph{aggregation} that can be formalised using \emph{monoid homomorphisms}. As we shall show, this not only generalises $\forall$ and~$\exists$, but also allows defining instrumentation operators for the aggregations \verb!\sum!, \verb!\product!, \verb!\numof!, \verb!\max!, and \verb!\min!. 
Such aggregation is supported in the specification languages JML~\cite{DBLP:books/daglib/p/LeavensBR99} and ACSL~\cite{acsl} in the form of \emph{extended quantifiers}, and is frequently needed for the specification of functional correctness properties.  Although commonly available in specification languages, most verification tools do not support the verification of properties specified using aggregation. Such properties instead typically have to be manually rewritten using standard quantifiers, pure recursive functions, or ghost code involving loops. This reduction step is error-prone, and represents an additional complication for automatic verification approaches, but can be elegantly handled using our proposed instrumentation framework.

\begin{definition}[Monoid]
  A \emph{monoid} is a structure~$(M, \circ, e)$ consisting of a non-empty set~$M$, a binary associative operation~$\circ$ on $M$, and a neutral element~$e \in M$.
  A \emph{monoid} is called \emph{commutative} if $\circ$ is commutative. 
\end{definition}

We model finite intervals of arrays as strings over the monoid~$(D^*, \cdot, \epsilon)$ of finite sequences over some data domain~$D$. The concatenation operator~$\cdot$ is non-commutative. The aggregation or quantification of finite intervals of an array is then defined as a monoid homomorphism from $(D^*, \cdot, \epsilon)$ to some target monoid.

\iffalse
\begin{lemma}[Partial Inverse, WIP]
    For any cancellative monoid, there exists an operation 
    $\circ^{-1}$ such that...
    $a \circ b \circ^{-1} a = b$, show equivalance to definitino of cancellativity above. 
\end{lemma}
\fi

\begin{definition}[Monoid Homomorphism]
  A \emph{monoid homomorphism} is a function
  $h : M_1 \to M_2$ between monoids~$(M_1, \circ_1, e_1)$ and $(M_2, \circ_2, e_2)$ such that $h(x \circ_1 y) = h(x) \circ_2 h(y)$ for all $x, y \in M_1$ and $h(e_1) = e_2$.
\end{definition}

% \subsection{Quantifiers as monoid homomorphisms} [rewrite and extend below]
For example, universal quantification can be modelled as a homomorphism to the monoid $(\mathbb{B}, \wedge, \mathit{true})$. 
A second example is the computation of the \emph{maximum} (similarly, \emph{minimum}) value in a sequence. For the domain of integers, the natural monoid to use is the algebra~$(\mathbbm{Z}_{-\infty}, \max, -\infty)$ of integers extended with $-\infty$\footnote{For machine integers, $-\infty$ could be replaced with \texttt{INT\_MIN}.}. 
% and the homomorphism~$h_{\max}$ is generated by mapping singleton sequences~$\langle n \rangle$ to the value~$n$.
%
A third example is the computation of the element \emph{sum} of an integer sequence, modelled by the homomorphism~$h_{\operatorname{sum}}$ to the target monoid $(\mathbbm{Z}, +, 0)$. 
The monoid in the last example, the computation of element sums, has the special property of being \emph{cancellative.}

\begin{definition}[Cancellative Monoid]
  A monoid~$(M, \circ, e)$ is \emph{cancellative} if one of the following equivalent properties hold:
  \begin{enumerate}
      \item $x \circ y = x \circ z$ implies $y = z$, and $y \circ x = z \circ x$ implies $y = z$, for all $x, y, z \in M$.
      \item there are partial operations  $\circ^{-1}_l, \circ^{-1}_r : M \times M \rightharpoonup M$ such that $(x \circ y) \circ^{-1}_r y = x$ and $x \circ^{-1}_l (x \circ y) = y$, for all $x, y \in M$.
  \end{enumerate}
\end{definition}

In the case of a \emph{commutative} cancellative monoid, it is not necessary to distinguish two separate functions~$\circ^{-1}_l, \circ^{-1}_r$ for left- and right-cancellation. We use the notation $(x \circ y) \circ^{-1} y$ in the commutative case.

Cancellative monoids generalise the notion of a group, and are a useful special case in the context of instrumentation: aggregation operators based on cancellative monoids can be represented using a particularly efficient kind of instrumentation. An example of a cancellative (commutative) monoid is the group~$(\mathbbm{Z}, +, 0)$ used to model \lstinline!\sum!, which calculates the sum of the elements in an array. 
% In \autoref{sec:cancellative}, we define a similar cancellative monoid (that is not a group) to model the operator \texttt{\textbackslash prod}\todo{JA: This operator is not for prod} for computing the product of the elements in an array.
%
%Similarly, the \emph{number of occurrences} of some element can be computed.

\subsubsection{Programming Language with Aggregation}
\label{sec:programming-language}

% [ToDo: Add some intro text to the section relating to previous section on normal quantifiers]
% We present our approach in the context of a core language, providing operations over integers, Booleans, and arrays. The distinguishing language feature is the \emph{aggregation} function on arrays. To introduce the language, we first recapitulate the notion of monoid homomorphisms, used later to formalize aggregation.
%
We extend our core
programming language a second time, this time adding support for aggregation of array values. We introduce the expressions
$\texttt{aggregate}_{M, h}\texttt{(}\nonTerm{Expr}\texttt{,} \nonTerm{Expr}\texttt{,} \nonTerm{Expr}\texttt{)}$,
and use monoid homomorphisms
to formalise them.
Recall that we use $D_{\sigma}$ to denote the domain of a program type~$\sigma$.
%
%With slight extensions (not discussed here), the setting could also cover index-dependent functions like $\mathit{indexOf}$.
%
% The most interesting feature of the language is
% We can now expand on
% the \texttt{aggregate} operator on arrays:
%

\begin{definition}\label{def:aggregation}
  Let~$\texttt{Array}~\sigma$ be an array type, $\sigma_M$ a program type,
  $M$~a commutative monoid that is a subset of $D_{\sigma_M}$, and $h : D_{\sigma}^* \to M$ a monoid homomorphism. 
  Let furthermore $\mathit{ar}$ be an expression of type $\texttt{Array}~\sigma$, and $l$ and $u$ be integer expressions. 
  Then, $\texttt{aggregate}_{M, h}\texttt{(}\mathit{ar}\texttt{,} l\texttt{,} u\texttt{)}$ is an expression of type~$\sigma_M$,
defined as:
\begin{equation}\label{eq:definition-aggregate}
    \llbracket \texttt{aggregate}_{M, h}\texttt{(}\mathit{ar}\texttt{,} l\texttt{,} u\texttt{)} \rrbracket_s
    ~=~
    h(\langle \llbracket \mathit{ar} \rrbracket_s (\llbracket l \rrbracket_s), \llbracket \mathit{ar} \rrbracket_s (\llbracket l \rrbracket_s + 1), \ldots, \llbracket \mathit{ar} \rrbracket_s (\llbracket u \rrbracket_s - 1) \rangle)
\end{equation}
% where $f = \llbracket a \rrbracket_s$, $l = \llbracket x \rrbracket_s$, and $u = \llbracket y \rrbracket_s$. 
\end{definition}

\begin{table}
    \centering
    \begin{tabular}{l@{\qquad}c@{\qquad}c}
      \toprule
      (Extended) Quantifier & Target Monoid & $h$ generated by \\
      \midrule
       \lstinline$\sum(ar,l,u)$ & $(\mathbb{Z},+,0)$ &  $\langle x \rangle \mapsto x$
       \\
       \lstinline$\min(ar,l,u)$ & $(\mathbb{Z}_{+\infty},\mathit{min},+\infty)$ &  $\langle x \rangle \mapsto x$
       \\
       \lstinline$\max(ar,l,u)$ & $(\mathbb{Z}_{-\infty},\mathit{max},-\infty)$ &  $\langle x \rangle \mapsto x$
       \\
       \lstinline$\product(ar,l,u)$ & $(\mathbb{Z},\cdot,1)$ &  $\langle x \rangle \mapsto x$
       \\
       \lstinline[mathescape]!\numof(ar,l,u,$\lambda$(x,i).P)! & $(\mathbb{Z},+,0)$ & $\langle b \rangle \mapsto b \;?\; 1 : 0$
       \\\midrule
       \lstinline[mathescape]!\forall(ar,l,u,$\lambda$(x,i).P)! & $(\mathbb{B},\wedge,\mathit{true})$ & $\langle (x, i)\rangle \mapsto P(x, i)$
       \\
       \lstinline[mathescape]!\exist(ar,l,u,$\lambda$(x,i).P)! & $(\mathbb{B},\vee,\mathit{false})$ & $\langle (x, i)\rangle \mapsto P(x, i)$
       \\\bottomrule
    \end{tabular}

    \smallskip
    \caption{Target monoid and homomorphism generator for all ACSL (extended) quantifiers\figstop{}}
    \label{tab:acsl_exquans_defs}
\end{table}

% \todo{DG: Do we need an equation here?}
%
Intuitively, the expression~$\texttt{aggregate}_{M, h}\texttt{(}\mathit{ar}\texttt{,} l\texttt{,} u\texttt{)}$ denotes the result of applying the homomorphism~$h$ to the slice~$\mathit{ar}[l \;..\; u-1]$ of the array~$\mathit{ar}$.
As a convention, in case $u < l$ we assume that the result of \texttt{aggregate} is $h(\langle\rangle)$.
As with array accesses, we assume also that \lstinline{aggregate} only occurs
in normalised statements of the form $\texttt{t = aggregate}_{M, h}\texttt{(}\mathit{ar}\texttt{,} l\texttt{,} u\texttt{)}$. 

We define aggregation operators corresponding to all extended quantifiers found in ACSL, with the homomorphism $h$ being uniquely defined by how it acts on singleton sequences. For example, for \lstinline!sum!, the homomorphism $h_{\mathit{sum}}$ is generated by the mapping $\langle n\rangle \mapsto n$. \autoref{tab:acsl_exquans_defs} shows, for each operator, the homomorphism and its domain monoid.    

\medskip
One may observe here that \lstinline!aggregate!, in a slightly extended form, subsumes the operator~\lstinline!quant!$_Q$ previously introduced for ordinary quantifiers. For this, consider a generalised \lstinline!aggregate'! operator that passes not only array elements, but also their indexes to a homomorphism:
\begin{equation*}
   \llbracket \texttt{aggregate'}_{M, h}\texttt{(}\mathit{ar}\texttt{,} l\texttt{,} u\texttt{)} \rrbracket_s
    ~=~
    h(\langle (\llbracket \mathit{ar} \rrbracket_s (\llbracket l \rrbracket_s), \llbracket l \rrbracket_s), \ldots, (\llbracket \mathit{ar} \rrbracket_s (\llbracket u \rrbracket_s - 1), \llbracket u \rrbracket_s - 1) \rangle)
\end{equation*}
For the quantifier~\lstinline!\forall!, we can then let $M_\forall$ be the monoid $(\mathbb{B}, \wedge, \mathit{true})$ and $h_\forall$ the homomorphism generated by mapping value-index pairs~$\langle(x, i)\rangle$ to the Boolean $P(x, i)$.

\subsubsection{An Instrumentation Operator for Non-Cancellative Monoids}
\label{sec:non-cancellative}
\begin{figure}[!htb]
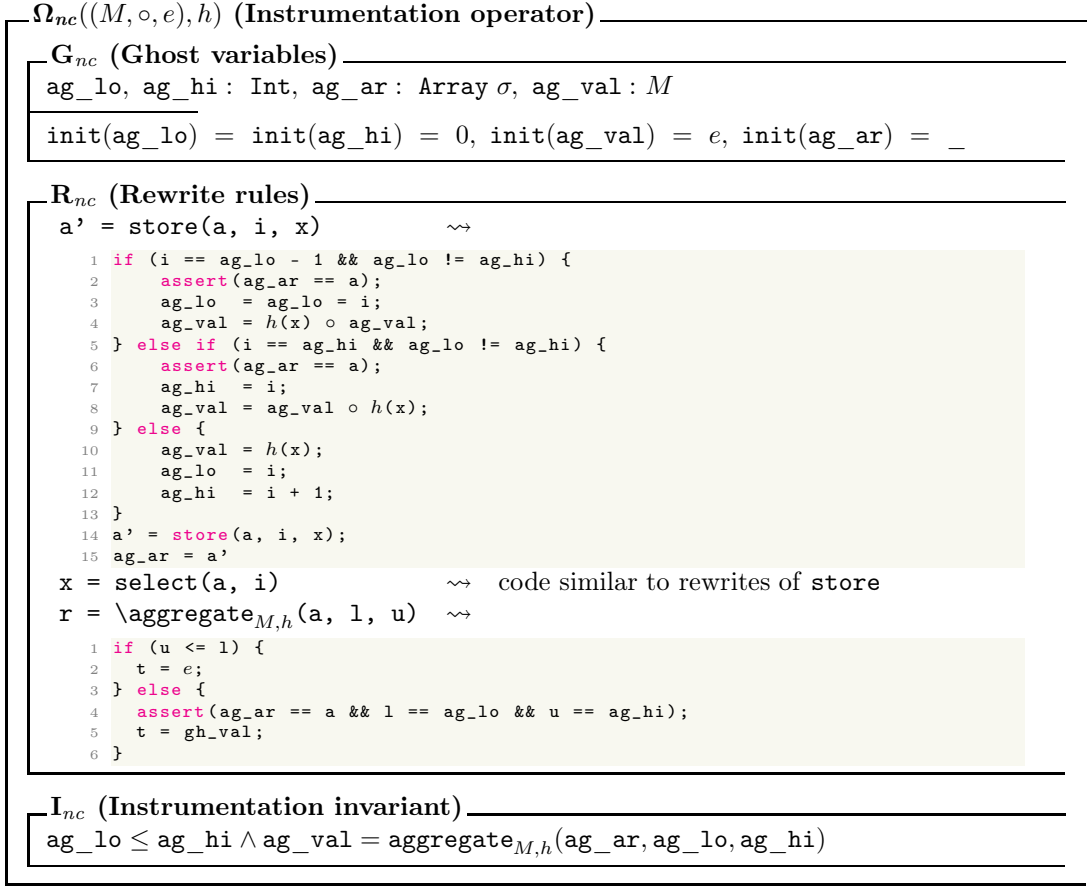

    \centering
\begin{minipage}{0.87\textwidth}
%\footnotesize
\begin{class}{\bm{\ncaninstropnoargs{}}((M,\circ,e), h) \textbf{
(Instrumentation operator)}}
\begin{schema}{\mathbf{G_{\mathit{nc}}} \textbf{ (Ghost variables)}}
    \mathtt{ag\_lo},~\mathtt{ag\_hi:~Int},
    \ \mathtt{ag\_ar:~Array\ }\sigma, 
    \ \mathtt{ag\_val:}~M
    \where
    \mathtt{init(ag\_lo)} ~=~
    \mathtt{init(ag\_hi)} ~=~ 0,\ 
    \mathtt{init(ag\_val)} ~=~ e,\ 
    \mathtt{init(ag\_ar)} ~=~ \_
\end{schema}
\\ %[-22pt]
\begin{schema}{\mathbf{R_{\mathit{nc}}} \textbf{ (Rewrite~rules)}}
\begin{array}{l@{\quad}l@{\quad}l@{\quad\qquad}l}
     \texttt{a' = store(a, i, x)}  &\leadsto&  & \\%\text{(R1)} \\
     \multicolumn{4}{l}{\hspace{2em}\mbox{
\begin{lstlisting}[linewidth=\opcodesnippetwidth, style=ExtWhileStyleOp]
if (i == ag_lo - 1 && ag_lo != ag_hi) {
    assert(ag_ar == a);
    ag_lo  = ag_lo = i;
    ag_val = $h$(x) $\circ$ ag_val;
} else if (i == ag_hi && ag_lo != ag_hi) {
    assert(ag_ar == a);
    ag_hi  = i;
    ag_val = ag_val $\circ$ $h$(x);
} else {
    ag_val = $h$(x);
    ag_lo  = i;
    ag_hi  = i + 1;
}
a' = store(a, i, x);
ag_ar = a'
\end{lstlisting}}}
    \\
     \texttt{x = select(a, i)}  &\leadsto&
        \text{code similar to rewrites of \texttt{store}} & %\text{(R2)}\\
    \\
     \texttt{r = \textbackslash }\texttt{aggregate}_{M,h}\texttt{(a, l, u)} &\leadsto&  & \\%\text{(R3)} \\
     \multicolumn{4}{l}{\hspace{2em}\mbox{
\begin{lstlisting}[linewidth=\opcodesnippetwidth, style=ExtWhileStyleOp]
if (u <= l) {
  t = $e$;
} else {
  assert(ag_ar == a && l == ag_lo && u == ag_hi);
  t = gh_val;
}
\end{lstlisting}}}\\
\end{array}
\end{schema}
\\ %[-22pt]
\begin{schema}{\mathbf{I_{\mathit{nc}}}\textbf{ (Instrumentation invariant)}}
     \mathtt{ag\_lo \leq ag\_hi \land ag\_val} = \mathtt{aggregate}_{M,h}\mathtt{(ag\_ar, ag\_lo, ag\_hi)}
     % \mathtt{ag\_lo = ag\_hi} \lor \mathtt{ag\_val = aggr(ag\_ar, ag\_lo, ag\_hi)}
\end{schema}
\end{class}
\caption{Instrumentation operator for aggregation parameterised on a non-cancellative target monoid $(M, \circ, e)$ and homomorphism $h$\figstop{}}
\label{fig:noncancel-instr-op}
\end{minipage}
\end{figure}
% A keen observer will have noticed that the definitions of $\Omega_{\forall, P}$ and $\Omega_{max}$ in Figures \autoref{fig:instr-op-forall} and \autoref{fig:max-instr-op} are extremely similar. Firstly, both operators have four ghost variables: The ghost variables \verb|gh_lo| and \verb|gh_hi| store the indices of the interval, \verb|gh_ar| stores the ghost array, and the last variable \verb|gh_max| and \verb|gh_P| for $\Omega_\mathit{max}$ and $\Omega_{\forall, P}$ respectively store the value of their respective (extended) quantifier on the ghost array between \verb|gh_lo| and \verb|gh_hi|. Furthermore, if we look at the rewrite rules $R_\mathit{max}$ and $R_{\forall, P}$, we see very similar behaviour: if we write far outide of the tracked inverval, we reset the tracking and if we write 
We are now ready to define instrumentation operators that work for any aggregation that is based on a monoid homomorphism. We start by considering the general case of monoids that are not (necessarily) cancellative, and define an instrumentation operator that is \emph{parameterised} on the target monoid $(M, \circ, e)$ and homomorphism $h: D_\sigma^* \rightarrow M$, as shown in \autoref{fig:noncancel-instr-op}.
The operator can be instantiated to obtain instrumentation operators for any ACSL extended quantifier, with the respective target monoid and homomorphism being as defined in \autoref{tab:acsl_exquans_defs}. 
In such instances of the parameterised instrumentation operator, $(M, \circ, e)$ and $h$ have to be replaced with their respective programming language implementation.
% providing implementations of $\circ$ and $h$. 

% syntactical constructs\todo{JA: Is "syntactical construct" good terminology?} $\bar{\circ}$ and $\bar{h}$ such that $\denot{\bar{\circ}} \defeq \circ$ and $\denot{\bar{h}} \defeq h$.

The operator $\ncaninstropnoargs{}$ is largely similar to the ones already defined for \lstinline!\max! and \lstinline!\forall!, wherefore we refer the reader to \autoref{sec:normal-quantifiers} for an explanation of the intuition behind the operator. The main difference is that the statements on lines 4 and~8 now use $\circ$ and $h$.

% For a monoid $(M, \circ, e)$, and homomorphism, $h$ can be found in \autoref{fig:noncancel-instr-op}. 

\begin{theorem}[Correctness of $\ncaninstropnoargs{}$]\label{theorem:correctness-non-cancellative}
    For any target monoid $(M, \circ, e)$ and monoid homomorphism $h: D_\sigma^* \rightarrow M$ with implementations of $\circ$ and $h$ that terminate and do not assign to any variables, 
    % in which the computation of $\denot{\circ}$\todo{JA: $h$ and $\circ$ are already in the semantic domain} and $\denot{h}$ do not assign to any variables and guaranteed to terminate, 
    the instrumentation operator $\ncaninstropgen{}$ is correct, i.e., it adheres to the constraints imposed by \autoref{def:instr-op}. 
\end{theorem}
\begin{proof}
Let $(M, \circ, e)$ and $h$ be such a target monoid and homomorphism.
% such that the computation of $\circ$ and $h$ do not assign to any variables and are guaranteed to terminate. 
We will show that the instrumentation operator created by implementing $\ncaninstropgen{}$ with terminating implementations of $\circ$ and $h$  that do not assign to any variables is correct.

First, observe that \lstinline{qu_lo} and \lstinline{qu_hi}
are initialised to the same value. Since the value of aggregation over an empty interval is the monoid identity~$e$, and \verb|ag_val| is initialised to $e$, the invariant $I$ is established by the initial values assigned by $G$.
It now remains to show that the rewrite rules adhere
to the remaining constraints. 
%We omit the case for
% \lstinline{select} statements as it is similar to
% \lstinline{store}.

We will first show that a rewritten \verb|store| terminates and satisfies the constraints in \autoref{def:instr-op}. It is easy to see that the instrumentation terminates, since no loop or recursion is introduced and the implementations of $\circ$ and $h$ are assumed to terminate. It is also easy to see that we only assign to ghost variables and \verb|a'|, and not to any other program variables, since $\circ$ and $h$ do not assign to any variables. Thus, the semantics is preserved, since the original \verb|store| statement remains unchanged on line 14 and the instrumentation code only assigns to ghost variables. Lastly, we must establish that the invariant is preserved, i.e., that the new values of \verb|ag_val|, \verb|ag_hi|, \verb|ag_lo|, and \verb|ag_ar| satisfy the invariant. For this, we distinguish three cases, each of which corresponds to a conditional block. In the first case, we assume that $\mathtt{i = ag\_lo - 1}$ and $\mathtt{ag\_lo \neq ag\_hi}$.  Since the value of $\mathtt{ag\_lo}$ is decreased, the first conjunct, $\mathtt{ag\_lo \leq ag\_hi}$, is preserved. The following derivation shows that the second conjunct of the invariant, $\mathtt{ag\_val=aggregate}_{M,h}\mathtt{(ag\_ar, ag\_lo - 1, ag\_hi)}$, is also preserved:
\begin{align*}
    \mathtt{ag\_val_{new}} & = h\mathtt{(x) \circ ag\_val_{old}} & 
    \\
     & = h\mathtt{(x)} \circ \mathtt{aggregate}_{M,h}\mathtt{(ag\_ar, ag\_lo, ag\_hi)} & \{\textit{By }I_{\mathit{nc}}\} 
     \\
     & = h\mathtt{(x)} \circ h\mathtt{(\left <ag\_ar[ag\_lo] , ~\ldots~ , ag\_ar[ag\_hi - 1] \right >)} & \{\textit{By \autoref{eq:definition-aggregate}}\} 
     \\ 
     & = h\mathtt{(\langle x \rangle \cdot \langle ag\_ar[ag\_lo], ~\ldots~ , ag\_ar[ag\_hi - 1]\rangle )} & \{h\> \textit{homomorphism}\}
     \\
     & = h\mathtt{(\left < ag\_ar[ag\_lo - 1] , ~\ldots~ , ag\_ar[ag\_hi - 1] \right >)} & \{\textit{$\mathtt{x = ag\_ar[i]}$} \>\mbox{and}
     \\
     && i = \mathtt{ag\_lo} - 1\}
     \}
     \\
     & = \mathtt{aggregate}_{M,h}\mathtt{(ag\_ar, ag\_lo - 1, ag\_hi)} & \{\textit{By \autoref{eq:definition-aggregate}}\}
\end{align*}
The second case, in which $\mathtt{i = gh\_hi}$  and $\mathtt{ag\_lo \neq ag\_hi}$, follows symmetrically. In the last case, in which $\mathtt{i}$ is different from $\mathtt{gh\_hi}$, $\mathtt{gh\_lo - 1}$, or it is the case that $\mathtt{ag\_lo = ag\_hi}$, we reset the ghost variables to track the singleton sequence. The invariant clearly holds afterwards. 

We omit the case for
\lstinline{select} statements as it is similar to
\lstinline{store}.

Lastly, the instrumentation code for \verb|aggregate| also satisfies the constraints in \autoref{def:instr-op}: \begin{inparaenum}[(a)]
    \item The instrumentation code obviously terminates.
    \item It only assigns to the variable \verb|t|.
    \item Since the ghost variables are not updated, the invariant is trivially preserved.
    \item To show that the semantics is preserved, we distinguish two cases. In the case \verb|u <= l|, we have defined the value of the aggregation \lstinline[mathescape]!aggregate$_{M,h}$(a, u, l)! as equal to \verb|e| in \autoref{sec:programming-language}, and therefore the semantics is preserved. In the other case, the value of \lstinline[mathescape]!aggregate$_{M,h}$(a, u, l)! is equal to \verb|ag_val| by the instrumentation invariant and therefore, by assigning \verb|ag_val| to \verb|t|, the semantics is preserved.
\end{inparaenum}
\end{proof}

\subsubsection{An Instrumentation Operator for Cancellative Monoids}
\label{sec:cancellative}
% Cancellative aggregation is aggregation based on a
% cancellative monoid.
When the target monoid is commutative and cancellative, we can track aggregate values faithfully even when storing
\emph{inside} of the tracked interval, unlike in the non-cancellative case.
% , unlike \lstinline{\max}.
% and universal quantification.
% where the tracking had to be reset when overwriting the
% previous maximum value.
Examples of cancellative operators are the aggregates \lstinline{\sum} and \lstinline!\numof!.
% with some differences which we will now explain.
%
\iffalse % Desc. of sum when op. was in app.
The sum of values in the interval would be tracked by the
instrumentation code, and when increasing the bounds of the tracked
interval, the new values are simply added to the tracked sum.
The main difference with \lstinline{\max}
occurs in the case when storing a value inside
of the tracked interval.
Here, the previous value at the index being written to
can be subtracted from the sum, before
the new value is added, enabling
the correct aggregate value to be computed.
\fi % Desc. of sum when op. was in app.

\begin{figure}[!htb]
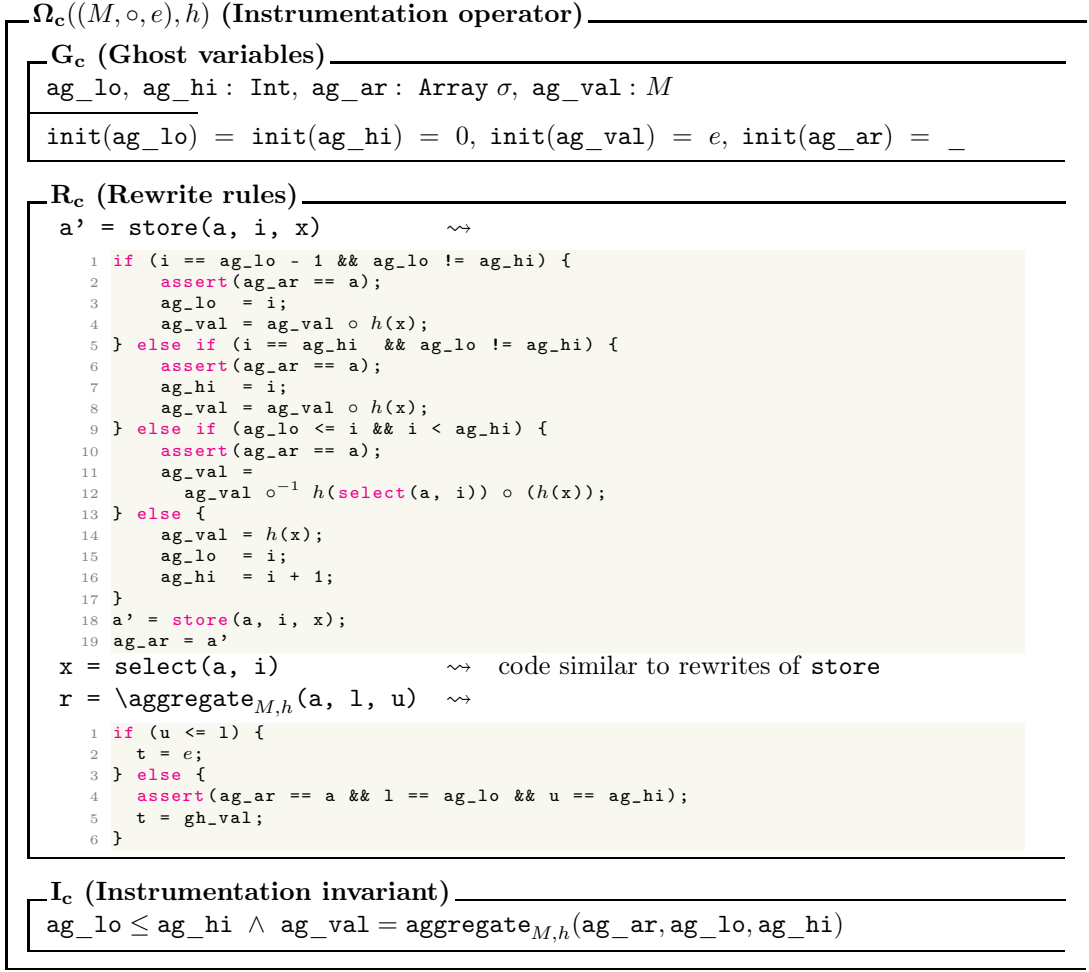

    \centering
    \begin{minipage}{0.87\textwidth}
%\footnotesize
\begin{class}{\mathbf{\Omega_c}((M,\circ,e),h) \textbf{
(Instrumentation operator)}}
\begin{schema}{\mathbf{G_c} \textbf{ (Ghost variables)}}
    \mathtt{ag\_lo},~\mathtt{ag\_hi:~Int},
    \ \mathtt{ag\_ar:~Array\ }\sigma, 
    \ \mathtt{ag\_val:}~M
    \where
    \mathtt{init(ag\_lo)} ~=~
    \mathtt{init(ag\_hi)} ~=~ 0,\ 
    \mathtt{init(ag\_val)} ~=~ e,\ 
    \mathtt{init(ag\_ar)} ~=~ \_
\end{schema}
\\ %[-22pt]
\begin{schema}{\mathbf{R_c} \textbf{ (Rewrite~rules)}}
\begin{array}{l@{\quad}l@{\quad}l@{\quad\qquad}l}
     \texttt{a' = store(a, i, x)}  &\leadsto&  & \\%\text{(R1)} \\
     \multicolumn{4}{l}{\hspace{2em}\mbox{
\begin{lstlisting}[linewidth=\opcodesnippetwidth, style=ExtWhileStyleOp]
if (i == ag_lo - 1 && ag_lo != ag_hi) {
    assert(ag_ar == a);
    ag_lo  = i;
    ag_val = ag_val $\circ$ $h$(x);
} else if (i == ag_hi  && ag_lo != ag_hi) {
    assert(ag_ar == a);
    ag_hi  = i;
    ag_val = ag_val $\circ$ $h$(x);
} else if (ag_lo <= i && i < ag_hi) {
    assert(ag_ar == a);
    ag_val = 
      ag_val $\smash{\circ^{-1}}$ $h$(select(a, i)) $\circ$ ($h$(x));
} else {
    ag_val = $h$(x);
    ag_lo  = i;
    ag_hi  = i + 1;
}
a' = store(a, i, x);
ag_ar = a'


\end{lstlisting}}}
    \\
     \texttt{x = select(a, i)}  &\leadsto&
        \text{code similar to rewrites of \texttt{store}} & %\text{(R2)}\\
    \\
    \texttt{r = \textbackslash }\mathtt{aggregate}_{M,h}\texttt{(a, l, u)} &\leadsto&  & \\%\text{(R3)} 
     \multicolumn{4}{l}{\hspace{2em}\mbox{
\begin{lstlisting}[linewidth=\opcodesnippetwidth, style=ExtWhileStyleOp]
if (u <= l) {
  t = $e$;
} else {
  assert(ag_ar == a && l == ag_lo && u == ag_hi);
  t = gh_val;
}
\end{lstlisting}}}\\
\end{array}
\end{schema}
\\ %[-22pt]
\begin{schema}{\mathbf{I_c}\textbf{ (Instrumentation invariant)}}
     \mathtt{ag\_lo \leq ag\_hi ~\land~}  \mathtt{ag\_val} = \mathtt{aggregate}_{M,h}\mathtt{(ag\_ar, ag\_lo, ag\_hi)}
\end{schema}
\end{class}
\caption{Instrumentation operator for aggregation parameterised on a \emph{cancellative} target monoid $(M, \circ, e)$ and homomorphism $h$\figstop{}}
    \label{fig:cancel-op}
    \end{minipage}
\end{figure}

To this end, we have defined the instrumentation operator $\caninstropnoargs{}$ in \autoref{fig:cancel-op}. This instrumentation operator can be instantiated with a cancellative monoid $(M, \circ, e)$ and a homomorphism $h$. 
% It is also parameterised on a function $g$ which will be the identity function in this section but will be elaborated upon in \autoref{sec:}.
% 
% 
% The definition of the instrumentation operator $\Omega_{\left (M, \circ, e \right), h, g}$, which is parameterised on a cancellative monoid $(M, \circ, e)$, a homomorphism $h$ and a function $g: M \rightarrow D$, where $D$ is the range of the aggregation operator, can be found
% in \autoref{fig:sum-instr-op}.
\iftrue % New desc. of sum when not in app.
% The instrumentation code tracks the value aggregation value in the interval
% and, when increasing the bounds of the tracked interval,
% the aggregation value is simply updated by performing the monoid operation on the aggregation value and the new value. 
The operator differs from $\ncaninstropnoargs{}$ in the rewrite rule for store operations, by adding a new \lstinline!else if! branch (lines 9--12). 
Since the monoid is cancellative and commutative, when storing inside of the tracked interval, 
we first remove 
the current value at the index
being written to 
% is first removed from the aggregation value, 
using $\circ^{-1}$,
and then add the new value using~$\circ$. Instrumentation operators for \lstinline{\sum} and \lstinline{\numof} can be defined in terms of $\Omega_c$, using the definitions in \autoref{tab:acsl_exquans_defs}.
The following result establishes correctness.

% ensuring that the correct
% aggregate value is computed.
\fi % New desc. of sum when not in app.
% The proof of the following correctness result is similar
% (and even simpler) to the one of \autoref{lem:max-correctness}:
%The following correctness result is proved in
%\iftr Appendix~\ref{app:correctnessSum}\else \cite{techreport}\fi.
% can be verified similarly to the one of \autoref{lem:max-correctness}:
%
\begin{theorem}[Correctness of $\caninstropnoargs{}$]
\label{thm:can-correctness}
$\caninstropgen{}$ is an instrumentation operator, i.e., it
adheres to the constraints imposed in \autoref{def:instr-op}.
\end{theorem}
% \todo[inline]{MV: We need to prove two things: (i)
% The instrumentation adheres to the constraints of \autoref{def:instr-op-correctness}; (ii) Some sort of monoid-related (homomorphism?) proof that states the equivalence of the definition of the aggregation operator over the domain of sequences (see \autoref{eq:definition-aggregate}) and the result of the instrumentation operation.}
\begin{proof}
    Similar to the proof of \autoref{theorem:correctness-non-cancellative}.
\end{proof}

\begin{corollary}\label{corollary:correctness-instr-acsl}
   There are correct instrumentation operators for all ACSL (extended) quantifiers.
\end{corollary}
\begin{proof}
The non-cancellative cases (\lstinline!\forall!, \lstinline!\exists!, \lstinline!\product!, \lstinline!\min!, \lstinline!\max!) follow from \autoref{theorem:correctness-non-cancellative}, by instantiating $(M, \circ, e)$ and $h$ according to \autoref{tab:acsl_exquans_defs}. The cancellative cases (\lstinline!\sum!, \lstinline!\numof!) follow similarly from \autoref{thm:can-correctness}.
\end{proof}

It is worth pointing out that we view $\ncaninstropnoargs{}$ and $\caninstropnoargs{}$ as a \emph{base recipe} for how to implement instrumentation operators for aggregation, but that one can further \emph{tailor} the instrumentation to specific aggregations. For example, the instrumentation operator for \lstinline!\forall! in \autoref{fig:instr-op-forall} specialises the rewriting of the aggregation, utilising that if a property holds for some portion of the array, then it also holds for a subset of that portion.

\subsubsection{Cancelling the Non-cancellative}\label{sec:cancel-non-cancel}
Up until now, we have been using the same monoid to define the instrumentation operator, henceforth denoted as $(M_{\mathit{aggr}},\circ_{\mathit{aggr}},e_{\mathit{aggr}})$, as we used to define the semantics of the extended quantifier itself. In this section, we create instrumentation operators using a different, \emph{cancellative} monoid, denoted as $(M_{\mathit{instr}},\circ_{\mathit{instr}},e_{\mathit{instr}})$, and parameterise our instrumentation operator on another homomorphism $g: M_{\mathit{instr}} \rightarrow M_{\mathit{aggr}}$ to map from the instrumentation monoid to the aggregation monoid when instrumenting an assignment using \verb|aggregate|. This means we slightly alter the definition of $\Omega_c$ in \autoref{fig:cancel-op} to be parameterised on a homomorphism $g$ and we change line 5 of the rewrite rule for \lstinline[mathescape]!r = aggregate$_{M,h}$(a, u, l)! to \verb|t = g(ag_val)|. 

% Although the monoids we defined in \autoref{sec:instrumentation-models} for existential and universal quantification were non-cancellative, we can leverage the definition of $\Omega_{numof, P}$ to create instrumentation operators which also support replacement. 
% 
We can use this technique to create a cancellative instrumentation operator for universal and existential quantification, instead of the non-cancellative operator we defined in \autoref{sec:non-cancellative}. By noting that an existential quantifier for a predicate~$P$ is true for an interval if the number of elements satisfying~$P$ in that interval is larger than~$0$, we can reuse the cancellative monoid we defined for \verb|\numof|. Whenever we instrument an assignment using \verb|\exists|, we use the homomorphism $g_{\exists,P}$ (see \autoref{eq:g-exists}) to map the number of elements satisfying the predicate to the result of the existential quantifier. This now allows us to define a stronger instrumentation operator for existential quantification as in \autoref{eq:stonger-instr-exists}.
%
% and as such, we can track the number of elements satisfying a certain predicate and convert this to a boolean result in the end. This allows us to define an alternative an instrumentation for existential quantification in terms of our cancellative monoid for \verb|numof|. 
% We can therefore define the instrumentation operator $\Omega_{\exists, P}$ as seen below and similar techniques can be used to create a cancellative instrumentation operator for universal quantification.
\begin{align}
    g_{\exists, P}(c) & \defeq \twopartdef{\mathsf{true}~~}{c > 0}{\mathsf{false}} \label{eq:g-exists}\\
    \Omega_{\exists, P} & \defeq \Omega_c({\left (\mathbbm{Z}, +, 0\right ),\, h_{\mathit{numof},P},\, g_{\exists, P}}) \label{eq:stonger-instr-exists}
\end{align}
A similar technique can be used to create a cancellative operator for \lstinline{\product}.
By noting that multiplication is cancellative over integers if we explicitly exclude $0$, it makes sense to define a monoid for multiplication in which the elements are tuples $( p, c )$ where $p$ is the product of the non-zero elements and $c$ is the number of elements equal to zero. This cancellative monoid can then be used to define $\Omega_{\mathit{prod}}$:
\begin{align*}
% Todo: how to represent the non-zero integers? 
    M_{\mathit{prod}} & \defeq \set{\left (p, c \right ) \mid p \in \mathbbm{Z}^{\neq 0},\ c \in \mathbbm{N}} \\
    \left (p_1, c_1 \right ) \circ_{\mathit{prod}} \left (p_2, c_2 \right ) & \defeq \left ( p_1 \cdot p_2, c_1 + c_2 \right ) \\
    e_{\mathit{prod}} & \defeq \left ( 1, 0 \right ) \\
    h_{\mathit{prod}}~\text{\small{homomorphism generated by }}\left<x\right> & \mapsto \twopartdef{\left ( x, 0 \right ) }{x \neq 0 }{\left ( 1, 1 \right )} \\
    g_{\mathit{prod}}(\left (p, c\right )) & \defeq  \twopartdef{p~~}{c = 0}{0} \\
    \Omega_{\mathit{prod}} & = \Omega_c({(M_{\mathit{prod}}, \circ_{\mathit{prod}}, e_{\mathit{prod}}),\, h_{\mathit{prod}},\, g_{\mathit{prod}}})
\end{align*}

% \subsection{Possible Extensions [??]}
% E.g., combining operators, or tracking multiple arrays.

%%%
\subsection{Deductive Verification of Instrumentation Operators}
As stated in \autoref{sec:instrumentation_operators},
instrumentation operators may be verified independently
of the programs to be instrumented.
The operators described in this article, i.e.,
operators for the (extended) quantifiers listed in \autoref{tab:acsl_exquans_defs} and the operator in \autoref{fig:squareInstrumentation},
%square, universal quantification, maximum, and sum,
have been verified in the verification tool
Frama-C~\cite{DBLP:conf/sefm/CuoqKKPSY12}.
The verification checks that the operators adhere to the constraints listed in \autoref{def:instr-op}.
The verified instrumentations are adaptations for the
C~language semantics and execution model.
More specifically, the adapted operators assume C native arrays,
rather than functional ones.

%\todo[inline]{Expand on the verification in Frama-C.}

% \subsection{Array operator extensions}
% \todo[inline]{Add some text about implemented (and potential) alterations to the above described operators. For example, the possibility to track several ranges.}

%%%%%%%%%%%%%%%%%%%%%%%%%%%%%%%%%%%%%%%%%%%%%%%%%%%%%%%%%%
%%%%%%%%%%%%%%%%%%%%%%%%%%%%%%%%%%%%%%%%%%%%%%%%%%%%%%%%%%

%%%%%%%%%%%%%%%%%%%%%%%%%%%%%%%%%%%%%%%%%%%%%%%%%%%%%%%%%%%%%%%%
\iftrue % combining instr. ops.
\section{Application of Multiple Instrumentation Operators}
\label{sec:operator-composition}

% \paragraph{Multiple Kinds of Aggregation Operators }
% \todo{CL: will move to sec. 2 and improve, and prove correctness
% % move "multiple kinds..." paragraph only, and rename to "combining instr. ops." or similar.
% }
The instrumentation operators considered so far have instrumented programs for a single, specific type of property, such as the $\maxinstrop{}$ operator for programs containing the \lstinline{\max} extended quantifier. In many cases this is sufficient, and the appropriate operator can be selected manually or by some simple syntactic analysis of the program to be verified. It is also common, however, that programs have to be rewritten multiple times to obtain an instrumented program that can be verified easily. This happens, for instance, whenever a program contains occurrences of multiple different extended quantifiers, or multiple occurrences of the same quantifier.

There are several ways to handle such cases in our framework. Naturally, instrumentation can be iterated: starting with a program~$P = P_0$, a chain~$P_0, P_1, P_2, \ldots, P_k$ can be constructed in which each $P_{i+1}$ is obtained by applying some instrumentation operator to $P_i$. If the final program~$P_k$ can successfully be verified, the correctness of $P = P_0$ follows.

Multiple instrumentation operators can also be composed to create new operators, which can instrument a program for many different properties, or many different occurrences of the same type of property, at once.

\iffalse
\newcommand{\opcomp}{\ensuremath{\cup}}
\else
\newcommand{\opcomp}{\ensuremath{\parallel}}
\fi
\begin{definition}[Composition of Instrumentation Operators]\label{def:compose}
    Two instrumentation operators $\Omega = (G, R, I)$ and $\Omega' = (G', R', I')$, with distinct ghost variables, can be composed to create a new instrumentation operator $\Omega'' = \Omega \opcomp \Omega'$:
    \begin{equation}
        \Omega' \opcomp \Omega \defeq (G \cup G', R \cup R', I \land I')
    \end{equation}
\end{definition}

The operators having distinct ghost variables means that, if $\mathit{vars}(G)$ is a function returning all the variable names in $G$ (i.e., the first element of each tuple), then $\mathit{vars}(G) \cap \mathit{vars}(G') = \varnothing$.
% (i.e. $G \cap G' = \emptyset$)
Note that it is possible to compose two operators of the same type, as long as the ghost variable names are substituted with fresh ones.

% \todo{JA: What breaks if the variables are instead assumed to be disjoint ("fresh"), or if $G''$ is the disjoint union of $G$ and $G'$?}
% \marginpar{MV: Unsure about the word composition and the symbol $\circ$. Also, maybe omit the proof and lemma?}
\begin{lemma}[Correctness of $\Omega \opcomp \Omega'$]
    If $\Omega = (G, R, I)$ and $\Omega' = (G', R', I')$ are both correct instrumentation operators (i.e., they adhere to \autoref{def:instr-op}), then the composed instrumentation operator $\Omega \opcomp \Omega'$ is also correct, i.e., it adheres to the conditions imposed in \autoref{def:instr-op}.
\end{lemma}
\begin{proof}
    Let $\Omega = (G, R, I)$, $\Omega' = (G', R', I')$ and $\Omega'' = \Omega \opcomp \Omega' = (G'', R'', I'')$. 

    % By \autoref{def:completeness}, $I$ and $I'$ only reason about ghost variables $G$ and $G'$ respectively and $G$ and $G'$ are distinct sets of variables that satisfy $I$ and $I'$ respectively, instrumentation invariants $I$ and $I'$ must both be by $G \cup G'$. We can therefore conclude that $I'' = I \land I'$ must also be satisfied by $G''$.

    By \autoref{def:instr-op}, $I$ and $I'$ are invariants over only the the variables in $G$ and $G'$, respectively.
    Since the sets of variables defined $G$ and $G'$ are disjoint, then it follows that their initial values also establish $I \land I'$.

    Since the rewrite rules are not changed, it follows that that they still terminate and assign only to the allowed variables.

    For the remaining two constraints, we consider the rewrite rules from $R$ (the proof is symmetrical for the rules from $R'$).
    Since the rules in $R$ preserve the invariant $I$ and we start from a stronger invariant $I \land I'$, it follows $I$ must hold afterwards. Since the rules in $R$ do not mention the variables in $G'$, $I'$ must also still hold, and $I \land I'$ is preserved.
    Similarly, since the rules in $R$ preserve the semantics of the original statement under $I$, it must also do so under the stronger $I \land I'$.
\end{proof}

% \marginpar{JA: What if $I$ contradicts $I'$? And what if $R'$ rewrites something from $R$?}
% \marginpar{Maybe $\parallel$ instead of $\circ$}
\fi % combining instr. ops.
%%%%%%%%%%%%%%%%%%%%%%%%%%%%%%%%%%%%%%%%%%%%%%%%%%%%%%%%%%%%%%%%

\section{Correctness of Instrumentation Operators}
\label{sec:instrumentation-correctness}

This section explains how we can reason about the \emph{correctness} of instrumentation operators.
For this, in \autoref{subsec:instrumentation-soundness} we introduce the two notions of \emph{soundness} and \emph{weak completeness,} which are properties that hold by construction for every instrumentation operator.
In \autoref{subsec:instrumentation-completeness} we show that for certain classes of programs, the instrumentation operators are also \emph{complete} in a stronger sense.
% \todo[inline]{DG: Add some introductory text about soundness, weak completeness, and completeness. Point out that the first two are generic and true by construction.}

\subsection{Soundness and Weak Completeness}
\label{subsec:instrumentation-soundness}

% To reason about correctness of instrumentation operators,
% we will argue that if a witness exists for the instrumented
% program, then it can be translated into a witness for the
% non-instrumented program.
Verification of an instrumented program produces one of two
possible results: a \emph{witness} if verification is successful, or a \emph{counterexample} otherwise.
A witness consists of the inductive invariants needed to verify the
program, and is presented in the context of the programming language:
it is translated back from the back-end theory used by the
verification tool,
and is a formula over the program variables and the ghost variables
added during instrumentation.
A counterexample is an execution trace leading to a failing assertion.
% We say that an instrumentation operator is \emph{sound} if such
% witnesses can be translated into witnesses for the original program,
% i.e., inductive invariants over only the program variables enabling
% verification.

\begin{definition}[Soundness]
\label{def:instr-op-correctness}
% An instrumentation operator~$\Omega$ is called \emph{sound} 
% if for every program~$P$ and instrumented program~$P' \in \Omega (P)$,
% whenever there is a witness $S'$ for~$P'$, then 
% $S'$ can be translated into a witness $S$ for~$P$.
An instrumentation operator~$\Omega$ is called \emph{sound} 
if for every program~$P$ and instrumented program~$P' \in \Omega (P)$,
whenever there is an execution of~$P$ where some \texttt{assert} 
statement fails, then 
there also is an execution of~$P'$ where some 
\texttt{assert} statement fails.
\end{definition}

% A corollary of \autoref{def:instr-op-correctness} is that if an
% instrumentation operator is \emph{sound}, then witnesses for the
% instrumented program can be translated into witnesses for the original
% program, i.e., inductive invariants over only the program variables.
Equivalently, existence of a witness for an instrumented program
entails existence of a witness for the original program,
in the form of a set of inductive invariants solely over the program variables.
Notably, because of the semantics-preserving nature
of the rewrites under the instrumentation invariant, a witness for the
original program can be derived from one for the instrumented
program. % which entails soundness.
One such back-translation is to add the instrumentation invariant
as a conjunct to the original witness, and to add existential
quantifiers for the ghost variables.
%
% \begin{theorem}[Soundness]
% \label{thm:soundness}
% Every instrumentation operator~$\Omega$ is sound.
% \end{theorem}
% %
% The proof is in Appendix~\ref{app:intrCorrectness}.
% Furthermore, since we are able to extract the inductive invariants needed for
% verification of the original program, it would now also be possible
% to use them to verify the program using non-automatic tools,
% e.g., deductive verification tools, that require such invariants
% to be manually supplied.

% \paragraph{Example.}

% \todo[inline]{Update example to not refer to square op. Maybe use the op. in \autoref{fig:instr-op-forall} and instrumented program \autoref{fig:quantified_example_instr}.}

%\todo[inline]{CL: Using square example here again since we have it in intro. Example for forall is also in source and can be toggled on, please decide which you want.}
% change \iftrue to \iffalse to change example
\iftrue % square example
\begin{example}
To illustrate the back-translation,
we return to the instrumentation operator from
\autoref{fig:squareInstrumentation} and the example
program from \autoref{fig:ex-non-linear-arithmetic-code}.
The witness produced by our verification tool in this case
is the formula:
\begin{equation*}
\begin{split}
\mathtt{i} = \mathtt{x\_shad} \wedge
\mathtt{x\_sq} + \mathtt{x\_shad} = 2s\wedge
\mathtt{N} \geq \mathtt{i} \wedge
\mathtt{N} \geq \mathtt{1} \wedge
2\mathtt{s} \geq \mathtt{i} \wedge
\mathtt{i} \geq  0
\end{split}
\end{equation*}
After conjoining the instrumentation invariant
$\mathtt{x\_sq} = \mathtt{x\_shad}^2$
% we can substitute occurrences of $\mathtt{ii}$ and $\mathtt{is}$
and adding existential quantifiers, 
we obtain
an inductive invariant that is sufficient to verify the original program:
\begin{equation*}
\begin{split}
\exists x_\mathrm{sq}, x_\mathrm{shad}.\; (
\mathtt{i} = x_\mathrm{shad} \wedge
x_\mathrm{sq} + x_\mathrm{shad} = 2s \wedge~~~~~~~~~~
\\
\mathtt{N} \geq \mathtt{i} \wedge
\mathtt{N} \geq \mathtt{1} \wedge
2\mathtt{s} \geq \mathtt{i} \wedge
\mathtt{i} \geq  0
\wedge x_\mathrm{sq} = x_\mathrm{shad}^2)
\end{split}
\end{equation*}
\end{example}
\else % quantification example
\begin{example}
To illustrate the back-translation,
we return to the instrumentation operator from
\autoref{fig:instr-op-forall} and the example
program from \autoref{fig:quantified_example} and \autoref{fig:quantified_example_instr}.
The witness produced by our verification tool in this case
is the formula:
\begin{equation*}
\begin{split}
\mathtt{qu\_P} \wedge
\mathtt{qu\_ar} = \mathtt{a} \wedge
\mathtt{qu\_lo} = 0 \wedge
\mathtt{qu\_hi} \geq \mathtt{i} \wedge
\mathtt{N} \geq \mathtt{i} \wedge
\mathtt{N} \geq 1 \wedge 
\mathtt{i} \geq 0
\end{split}
\end{equation*}
After conjoining the instrumentation invariant
\begin{align*}
\mathtt{qu\_lo}& = \mathtt{qu\_hi} \; \lor \\
     &(\mathtt{qu\_lo} < \mathtt{qu\_hi}\land 
     \mathtt{qu\_P} = 
     \mathtt{forall(qu\_ar, qu\_lo, qu\_hi,} \lambda\mathtt{(x, i).P))}
\end{align*}
and existentially quantifying over the involved ghost variables, 
we obtain
an inductive invariant that is sufficient to verify the original program:
\begin{equation*}
\begin{split}
\exists \, q_\mathrm{ar}, q_\mathrm{lo}, q_\mathrm{hi}, q_\mathrm{P}.\; (
% witness
q_\mathrm{P} \wedge
q_\mathrm{ar} = \mathtt{a} \wedge
q_\mathrm{lo} = 0 \wedge
q_\mathrm{hi} \geq \mathtt{i} \wedge
\mathtt{N} \geq \mathtt{i} \wedge
\mathtt{N} \geq 1 \wedge 
\mathtt{i} \geq 0
\wedge~~~~~~~~~~
\\
%\wedge 
( % instr. inv.
    q_\mathrm{lo} = q_\mathrm{hi}
    \vee (q_\mathrm{lo} < q_\mathrm{hi} \wedge
    q_\mathrm{P} =  \mathtt{forall(q_\mathrm{ar}, q_\mathrm{lo}, q_\mathrm{hi},} \lambda\mathtt{(x, i).P))}
)
\end{split}
\end{equation*}
\end{example}
\fi

% \begin{equation*}
% \begin{split}
% \exists \mathtt{x\_sq}, \mathtt{x\_shad}.\; (
% \mathtt{i} = \mathtt{x\_shad} \wedge
% \mathtt{x\_sq} + \mathtt{x\_shad} = 2 \times s \wedge~~~~~~~~~~
% \\
% \mathtt{N} \geq \mathtt{i} \wedge
% \mathtt{N} \geq \mathtt{1} \wedge
% 2 \times \mathtt{s} \geq \mathtt{i} \wedge
% \mathtt{i} \geq  0
% \wedge \mathtt{x\_sq} = \mathtt{x\_shad}^2)
% \end{split}
% \end{equation*}
%
% \dg{Shouldn't the definition talk instead about failing \lstinline{assert} statements not added by the instrumentation?}
\begin{definition}[Weak Completeness]
\label{def:weak-completeness}
The operator~$\Omega$ is called \emph{weakly complete} 
if for every program~$P$ and instrumented program~$P' \in \Omega(P)$,
whenever an \texttt{assert} statement that
has not been added to the program by the instrumentation fails in the instrumented
program~$P'$, then it also fails in the original program~$P$.
% An instrumentation operator~$\Omega$ is called \emph{weakly complete} 
% if for every program~$P$ and instrumented program~$P' \in \Omega (P)$,
% whenever a counterexample $C'$ for an assertion in~$P'$ that has not been
% added by the instrumentation is found,
% then  $C'$ can be translated into a counterexample $C$ for~$P$.
\end{definition}
Similarly to the back-translation of invariants,
when verification fails, counterexamples
for assertions of the original program, found during
verification of the instrumented program, can be translated back to
counterexamples for the original program.
% which in turn entails weak completeness.
We thus obtain the following result.
% (proof given in Appendix~\ref{app:intrCorrectness}).
%
\begin{theorem}[Soundness and weak completeness]
% \label{thm:weak-completeness}
\label{thm:soundness-weak-completness}
Every instrumentation operator~$\Omega$ is sound and weakly complete.
\end{theorem}
\begin{proof}% of \autoref{thm:soundness-weak-completness}}
Let~\instropdef{} be an instrumentation operator.
Since $I$ is a formula over ghost variables only,
which holds initially and is preserved by all rewrites,
$I$~is an invariant of the fully instrumented program.
This entails that rewrites of assignments are semantics-preserving.
Furthermore, since instrumentation code only assigns to ghost variables
or to $r$ (i.e., the left-hand side of the original
statement), program variables have the same valuation
in the instrumented program as in the original one.
Furthermore, since all rewrites are terminating under~$I$, the instrumented program will terminate if and only if the original
program does.

In the case when verification succeeds, and a witness
is produced, weak completeness follows vacuously.
A witness consists of the inductive invariants sufficient to verify
the instrumented program.
Thus, they are also sufficient to verify the
assertions existing in the original program,
since assertions are not rewritten and all program variables
have the same valuation in the original and the instrumented programs.
%
% Since all rewrites are semantics-preserving under~$I$,
% any inductive invariant of the instrumented program in
% conjunction with~$I$, taken as a formula with the
% ghost variables existentially quantified, is an inductive invariant of the original
% program.
% Thus, a witness for an instrumented program can be translated back to a
% witness for the original program.
% If a witness exists for the original program, whenever one
% exists for the instrumented program, then any failing assertions in
% the original program must also fail in the instrumented program,
% and the instrumentation operator is therefore sound.
Since a witness for the instrumented program can be back-translated
to a witness for the original program, any failing assertion
in the original program must also fail after instrumentation,
and $\Omega$ is therefore sound.

In the case when verification fails, soundness follows vacuously,
and if the failing assertion was added during instrumentation, also
weak completeness follows.
If the assertion existed in the original program, since such assertions
are not rewritten, and since program variables have the same valuation
in the instrumented program as in the original program,
then any counterexample for the instrumented program is also
a counterexample for the original program, when projected
onto the program variables.
% \qed
\end{proof}

%%%%%%%%%%%%%%%%%%%%%%%%%%%%%%%%%%%%%%%%%%%%%

\subsection{Completeness}
\label{subsec:instrumentation-completeness}

Weak completeness of an instrumentation operator only guarantees that instrumentation does not affect the validity of assertions that are left intact by the instrumentation.
A ``proper'' notion of completeness would state something considerably stronger, namely, that for any correct program, there is a program in the instrumentation space of the operator that is also correct -- with the intention that the instrumented program can be automatically verified. 
This latter aspect can be captured by restricting the operators or instructions that the instrumented program is allowed to contain. 

% \todo[inline]{DG: In the story below, if one wants, one can build in the desired effect of instrumentation by replacing $\Omega(P)$ with a subset $\hat{\Omega} (P)$ that only contains programs of the desired kind (for instance, without ``target'' statements of a given type).}

To formalise this notion of completeness,
let~$\mathcal{A}$ denote a class of programs; $\mathcal{A}$ could, for instance, contain all programs that do not use certain aggregation operators, and which will therefore be handled well by a given back-end solver.
For a program~$P$ and an instrumentation operator~$\Omega$, we then denote by $\Omega_\mathcal{A} (P) = \Omega(P) \cap \mathcal{A}$ the set of instrumented versions of~$P$ that belong to $\mathcal{A}$. 

\begin{definition}[$\mathcal{A}$-Relative Completeness]
\label{def:completeness-dg}
Let $\Omega$ be an instrumentation operator, $\mathcal{A}$~a class of programs, and $P$ a program. We say that \emph{$\Omega$~is complete for~$P$} and~$\mathcal{A}$ if and only if the correctness of~$P$ implies the correctness of some instrumented program~$P' \in \Omega_\mathcal{A} (P)$:
$$ \mathsf{correct}(P) \>\Rightarrow\> \exists P' \in \Omega_\mathcal{A} (P).\, \mathsf{correct}(P') $$
We denote with~$\mathcal{P}_{\Omega_\mathcal{A}}$ the class of programs for which $\Omega$ is complete for~$\mathcal{A}$. 
\end{definition}

Note that the above notion of $\mathcal{A}$-relative completeness assumes that a single application of an instrumentation operator suffices to eliminate, in the program~$P$, all problematic constructs of the type for which the operator is defined. 
This, of course, is a strong assumption. In many cases, a single application will not suffice, but a series of applications of the same operator may achieve the task. 
Fortunately, it is straightforward to generalise the definition of  $\mathcal{A}$-relative completeness for this case. 
%
% By noting that many rewrite rules of instrumentation operators leave the assignment statements that they rewrite unchanged and merely add instrumentation code, we realise that one can apply instrumentation operators to programs iteratively. 
%
Let $\Omega^\ast (P)$ denote the set of programs that result from applying the instrumentation operator~$\Omega$ to the program~$P$ any finite number of times in a sequence. One can then define $\Omega^\ast_\mathcal{A} (P) = \Omega^\ast (P) \cap \mathcal{A}$, and adapt \autoref{def:completeness-dg} accordingly. We do not lose any generality by iteratively applying the \emph{same} operator a finite number of times, since we can leverage \autoref{def:compose} to compose multiple, different operators into a single operator that subsumes its parts. 
% \todo[inline]{MV: Add a comment on composing operators.}

The notion of $\mathcal{A}$-relative completeness can only be of practical value if one can formulate \emph{conditions on programs} that are sufficient for a program to belong to the class~$\mathcal{P}_{\Omega_\mathcal{A}}$.
The approach we take here is to formulate, for concrete instrumentation operators~$\Omega$ and program classes~$\mathcal{A}$, conditions on programs~$P$ that ensure that there is a program~$P'$ in the instrumentation space~$\Omega_\mathcal{A} (P)$ in which none of the assertions added to~$P$ by the instrumentation fails. 
By virtue of weak completeness, $\Omega$~must then be complete for~$P$ and~$\mathcal{A}$.

\subsubsection{Completeness Conditions for Non-Cancellative Operators}

As the base case, we focus on the class~$\mathcal{A}_B$ of programs in our core  language  that do not contain any aggregation expressions. Completeness with respect to the class~$\mathcal{A}_B$ captures the ability of an operator to \emph{eliminate} aggregation expressions from a program; the resulting program can then be processed by many existing software verification tools, which may be unable to handle aggregation, but are otherwise fully automatic.

\begin{definition}[Class $\mathcal{P}_{nc}((M, \circ, e), h)$]
\label{def:prog-cond-noncanc}
For given target monoid $(M, \circ, e)$ and monoid homomorphism $h: D^*_\sigma \rightarrow M$, we define $\mathcal{P}_{nc}((M, \circ, e), h)$ as the class of programs~$P$ satisfying the following conditions:
\begin{enumerate}
\item The program~$P$ contains exactly one statement $\mathtt{r = \textbackslash aggregate_{M,h}(a, l, u)}$ involving aggregation, for an array variable~$\typeJudg{\texttt{a}}{\texttt{Array}~\sigma}$.
% \item There is a variable \verb|a|
% \item All store operations take the form of \verb|a = store(a, i, x)|
%\item There is a variable \verb|a| which stores an array.
\item In any execution of $P$, prior to aggregation the array elements \lstinline!a[i]! between \texttt{l} and $\texttt{u} - 1$  are either read using \verb|select(a, i)| or updated with \verb|a = store(a, i, x)|.
\item These selects and stores occur sequentially, starting at \texttt{l} and ending at $\texttt{u} - 1$.
\item No other accesses or assignments to \lstinline!a! occur in $P$.
\end{enumerate}
\end{definition}
\iffalse
\todo[inline]{PR: There is a problem with this definition: in our core language, arrays are functional, so it is ambiguous to say that we are reading/writing to a. \\
MV:  I've tried to address this but not 100\% happy with my formulation.}
\fi

It should be pointed out that the class~$\mathcal{P}_\mathit{nc}$ is a \emph{semantic} class; in fact, it is not decidable whether a given program belongs to $\mathcal{P}_\mathit{nc}$ or not, since conditions~2--4 cannot be checked effectively in the general case. The criteria are still useful to characterise cases in which the non-cancellative instrumentation operator can succeed. For instance, our motivating example in \autoref{fig:battery_example} satisfies our conditions, which tells us that we successfully apply the instrumentation operator $\maxinstrop{}$ to obtain an instrumented program without extended quantifiers.

\begin{theorem}[Completeness of $\mathcal{P}_{nc}((M, \circ, e), h)$]
\label{theorem:completness-non-cancellative-dg}
For a given target monoid $(M, \circ, e)$ and monoid homomorphism $h: D^*_\sigma \rightarrow M$, let $\Omega = \Omega_{nc}((M, \circ, e), h)$.
% be an instrumentation operator for a non-cancellative aggregation operation. 
%
Then, $\Omega$~is complete for every program in the class~$\mathcal{P}_\mathit{nc}((M, \circ, e),h)$ and~$\mathcal{A}_B$, i.e.,
$ \mathcal{P}_\mathit{nc}((M, \circ, e),h) \subseteq \mathcal{P}_{\Omega_{\mathcal{A}_B}}$.
\end{theorem}

\begin{proof}
Let $(M, \circ, e)$ be a target monoid, $h: D^*_\sigma \rightarrow M$ be a monoid homomorphism and $\Omega = \Omega_{nc}((M, \circ, e), h)$.
Let~$P$ be a program satisfying the conditions of Definition~\ref{def:prog-cond-noncanc}, i.e., let $P \in \mathcal{P}_\mathit{nc}$. 
Assume that~$P$ is correct, i.e., that $\mathsf{correct}(P)$ holds.

Let $P' \in \Omega (P)$ be the instrumented program in which the statement $\mathtt{r = \textbackslash aggregate_{M,h}(a, l, u)}$ and all select and store operations have been rewritten.
Clearly, $P' \in \mathcal{A}_B$, and hence $P' \in \Omega_{\mathcal{A}_B} (P)$.

% \todo[inline]{MV: Can we just pick $P' \in \Omega_\mathcal{A} (P)$? Should we not argue that such a $P' \in \Omega(P)$ and argue that this $P' \in \Omega_\mathcal{A} (P)$? \\
% DG: My proof just follows Definition~\ref{def:completeness-dg}.}

Next assume, for the sake of obtaining a contradiction, that $P'$ is not correct, i.e., that $\neg \mathsf{correct}(P')$ holds. 
By virtue of the definition of~$\Omega$ (see Figure~\ref{fig:noncancel-instr-op}), and by weak completeness of~$\Omega$, one of the following assertions, added by the instrumentation, must have failed:
\begin{enumerate}
\item an \texttt{assert(ag\_ar == a)} in an instrumentation of \texttt{store}, 
\item an \texttt{assert(ag\_ar == a)} in an instrumentation of \texttt{select},  or
\item an \texttt{assert(ag\_ar == a \&\& l == ag\_lo \&\& h == ag\_hi)}.
\end{enumerate}
The next step of the proof is a case analysis on each of the 5 conjuncts. We show here one case, as the other cases follow similarly.

Assume that the assertion in an instrumented \texttt{store} statement failed. This means that the ghost array and actual array must be unequal. Since the interval is non-empty whenever \verb|ag_ar == a|, this means that there must be some assignment to \texttt{a} that does not update the ghost variable. The former contradicts our choice of~$P'$ that says that all store and select operations to \verb|a| have been rewritten. Since the interval is non-empty and we only increase the size of the interval through \verb|store| or \verb|select| statements, the latter implies that there is an assignment to \verb!a! after a \verb|store| or \verb|select|, which contradicts condition 4.

Hence, $P'$ must be correct, i.e., $\mathsf{correct}(P')$ holds, and therefore, by Definition~\ref{def:completeness-dg}, $P \in \mathcal{P}_{\Omega_{\mathcal{A}_B}}$.
\end{proof}

\subsubsection{Completeness Conditions for Cancellative Operators}

We can give a similar characterisation as in \autoref{def:prog-cond-noncanc} also for cancellative operators. As in the previous section, let $\mathcal{A}_B$ be the programs that do not contain aggregation expressions.

\begin{definition}[Class $\mathcal{P}_{c}((M, \circ, e), h)$]
\label{def:prog-cond-canc}
For given cancellative and commutative target monoid $(M, \circ, e)$ and monoid homomorphism $h: D^*_\sigma \rightarrow M$, we define $\mathcal{P}_{c}((M, \circ, e), h)$ as the class of programs~$P$ satisfying the following conditions:
\begin{enumerate}
\item The program~$P$ contains exactly one statement $\mathtt{r = \textbackslash aggregate_{M,h}(a, l, u)}$ involving aggregation, for an array variable~$\typeJudg{\texttt{a}}{\texttt{Array}~\sigma}$.
\item In any execution of $P$, prior to aggregation, each index \verb|i| between \texttt{l} and $\texttt{u} - 1$ is either read using \verb|select(a, i)| or written to using \verb|a = store(a, i, x)|. We distinguish the first operation that writes to a specific index from all subsequent operations and call the first operation that writes to a specific index a \emph{frontier-expanding operation}.
\item The indices that are accessed by frontier-expanding operations occur sequentially, starting at index \texttt{l} and ending at $\texttt{u} - 1$.
\item All other accesses to 
\lstinline!a! only operate on indices which have already been accessed by frontier-expanding operations.
% \item There is exactly one control point such that $P(c) \in T$.
% \item Each position in the array $\mathtt{a}$ between \texttt{low} and \texttt{high - 1} is written to or read exactly once before $T$ is executed.
% \item The reads and writes to a mentioned above occur sequentially, starting at index \texttt{low} and ending at \texttt{high - 1}.
\end{enumerate}
\end{definition}

\begin{theorem}[Completeness of $\mathcal{P}_{c}((M, \circ, e), h)$]
\label{theorem:completness-cancellative-dg}
For given cancellative and commutative target monoid $(M, \circ, e)$ and monoid homomorphism $h: D^*_\sigma \rightarrow M$, let $\Omega = \Omega_{c}((M, \circ, e), h)$.
% Let $\Omega$ be an instrumentation operator for a cancellative aggregation operation. 
%
Then, $\Omega$~is complete for the class~$\mathcal{P}_\mathit{c}((M, \circ, e), h)$ and~$\mathcal{A}_B$, i.e., 
$ \mathcal{P}_\mathit{c}((M, \circ, e), h) \subseteq \mathcal{P}_{\Omega_{\mathcal{A}_B}} $.
% Let $P \in \mathsf{Prog}$ and $c \in C_P$,  and $T$ be the set of statements of the form \texttt{m = aggr(a, low, high)}, where \texttt{a} is an array of numbers and \texttt{low} and \texttt{high} are integers. 
% The instrumentation operator $\Omega_{max}$ is complete for all  programs~$\pcancel$. 
\end{theorem}

\begin{proof}
The proof of \autoref{theorem:completness-cancellative-dg} follows closely that of \autoref{theorem:completness-non-cancellative-dg}.
\end{proof}

\section{Instrumentation Application Strategies}
\label{sec:transformation-strategies}

% \begin{figure}[tb]
%         \lstinputlisting[style=CStyle]{examples/strategy_example_1.c}
%         \caption{Example program operating on two arrays}
%         \label{fig:strategy_example_1}
% \end{figure}

%\begin{figure}[tb]
%        \lstinputlisting[style=ExtWhileStyle]{examples/strategy_example2_sum.c}
%        \caption{Example program operating on two arrays}
%        \label{fig:strategy_example_1}
%\end{figure}

\begin{figure}[tb]
\centering
\begin{minipage}{0.87\textwidth}
\begin{algorithm}[H]
  \SetKwInOut{Input}{Input}
  
  \Input{Program~$P$; statements~$S$; instrumentation space~$R$; oracle~$\mathsf{correct}$.}
  \KwResult{Instrumentation~$r \in R$ with $\mathsf{correct}(P_r)$; $\mathit{Incorrect}$; or $\mathit{Inconclusive}$.}
  
  \Begin{
  $\mathit{Cand} \leftarrow R$\;
  \While{$\mathit{Cand} \not= \emptyset$}{
    pick $r \in \mathit{Cand}$\;
    \uIf{$\mathsf{correct}(P_r)$}{
        \Return{$r$}\;
    }
    \Else{
        $\mathit{cex} \leftarrow $ counterexample path for $P_r$\;
%        $\mathit{failingStmt} \leftarrow $ failing assertion in $\mathit{cex}$\;
        \uIf{failing assertion in $\mathit{cex}$ also exists in $P$}{
            \tcc{$\mathit{cex}$ is also a counterexample for $P$}
            \Return{$\mathit{Incorrect}$}\;
        }
        \Else{
            \tcc{instrumentation on $\mathit{cex}$ may have been incorrect}
            $C' \leftarrow \{ p \in C \mid \mathit{ins}_r(p) \text{~occurs on~} \mathit{cex} \}$\;
            $\mathit{Cand} \leftarrow \mathit{Cand} \setminus \{ r' \in \mathit{Cand} \mid r(s) = r'(s) \text{~for all~} p \in C' \}$\;
        }
    }
  }
  \Return{$\mathit{Inconclusive}$}\;
  } 
  \caption{Counterexample-guided instrumentation search\figstop{}}
  \label{alg:instrumentation_search}
\end{algorithm}
\end{minipage}
\end{figure}

We will now define a counterexample-guided search procedure to discover applications of instrumentation operators that make it possible to verify a program.
%For presentation, we assume that only a single instrumentation operator is considered, although the algorithm directly generalizes to the simultaneous application of multiple operators. 

For our algorithm, we assume again that we are given an oracle~$\mathsf{correct}$ that is able to check the correctness of programs after instrumentation. Such an oracle could be approximated, for instance, using a software model checker. The oracle is free to ignore all functions, like aggregation, we are trying to eliminate by instrumentation; for instance, in \autoref{fig:ex-non-linear-arithmetic-code}, the oracle can over-approximate the term~\lstinline!N*N! by assuming that it can have any value. We further assume that $C$ is the set of control points of a program~$P$ corresponding to the statements to which a given set of instrumentation operators can be applied. For each control point~$p \in C$, let $Q(p)$ be the set of rewrite rules applicable to the statement at $p$, including also a distinguished value~$\bot$ that expresses that $p$ is not modified. For the program in \autoref{fig:ex-non-linear-arithmetic-code}, for instance, the choices could be defined by $Q(\mathtt{A}) = Q(\mathtt{B}) = \{ \mathrm{(R1)} , \bot \}$, $Q(\mathtt{C}) = \{ \mathrm{(R2)} , \bot \}$, and $Q(\mathtt{D}) = \{ \mathrm{(R4)} , \bot \}$, referring to the rules in \autoref{fig:squareInstrumentation}. Any function~$r : C \to \bigcup_{p \in C} Q(p)$ with $r(p) \in Q(p)$ will then define one possible program instrumentation. We will denote the set of well-typed functions~$C \to \bigcup_{p \in C} Q(p)$ by $R$, and the program obtained by rewriting~$P$ according to $r \in R$ by $P_r$. We further denote the control point in $P_r$ corresponding to some $p \in C$ in $P$ by $\mathit{ins}_r (p)$.

\autoref{alg:instrumentation_search} presents our algorithm to search for instrumentations that are sufficient to verify a program~$P$. The algorithm maintains a set~$\mathit{Cand} \subseteq R$ of remaining ways to instrument~$P$, and in each loop considers one of the remaining elements~$r \in \mathit{Cand}$ (line~4). If the oracle manages to verify $P_r$ in line~5, due to soundness of instrumentation the correctness of $P$ has been shown (line~6); if $P_r$ is incorrect, there has to be a counterexample ending with a failing assertion (line~8). There are two possible causes of assertion failures: if the failing assertion in $P_r$ already existed in $P$, then due to the weak completeness of instrumentation also $P$ has to be incorrect (line~10). Otherwise, the program instrumentation has to be refined, and for this from $\mathit{Cand}$ we remove all instrumentations~$r'$ that agree with $r$ regarding the instrumentation of the statements occurring in the counterexample (line~13).

Since $R$ is finite, and at least one element of $\mathit{Cand}$ is eliminated in each iteration, the refinement loop terminates. The set~$\mathit{Cand}$ can be exponentially big, however, and therefore should be represented symbolically; using BDDs, or using an SMT solver managing the set of blocking constraints from line~13.

We can observe soundness and completeness of the algorithm w.r.t.\ the considered instrumentation operators. 
%(proof in \iftr Appendix~\ref{app:instrSearchCorrectness}\else \cite{techreport}\fi):
\begin{lemma}[Correctness of \autoref{alg:instrumentation_search}]
    \label{lem:searchCorrectness}
    \autoref{alg:instrumentation_search} is sound and, in a certain sense, complete:
    \begin{enumerate}
        \item 
    If \autoref{alg:instrumentation_search} returns an instrumentation~$r \in R$, then $P_r$ and $P$ are correct.
    \item If \autoref{alg:instrumentation_search} returns $\mathit{Incorrect}$, then $P$ is incorrect.
    \item If there is $r \in R$ such that $P_r$ is correct, then \autoref{alg:instrumentation_search} will return $r'$ such that $P_{r'}$ is correct.
    \end{enumerate}
\end{lemma}

% \paragraph{Proof of \autoref{lem:searchCorrectness}}
\begin{proof}
\autoref{alg:instrumentation_search} will return an instrumentation~$r$
of when it has derived that $P_r$ is correct; due to the soundness
of instrumentation operators, then also $P$ is correct.

\autoref{alg:instrumentation_search} will return $\mathit{Incorrect}$ only
when it has discovered a counterexample for $P_r$ that ends in a failing assertion
that also occurs in $P$, i.e., that has not been introduced as a part of
instrumentation. Due to the weak completeness of instrumentation operators,
then also $P$ is incorrect.

Assuming that there is an $r$ such that $P_r$ is correct, the correctness
also of $P$ follows. \autoref{alg:instrumentation_search} can then not
return the result $\mathit{Incorrect}$. To see that \autoref{alg:instrumentation_search}
will eventually find some instrumentation~$r'$ such that $P_{r'}$ is correct, note
that the algorithm will in lines~12--13 only eliminate instrumentations~$r''$
such that $P_{r''}$ is incorrect.
% \qed
\end{proof}

%%%%%%%%%%%%%%%%%%%%%%%%%%%%%%%%%%%%%%%%%%%%%%%%%%%%%%%%%%%%%%%%
%%%%%%%%%%%%%%%%%%%%%%%%%%%%%%%%%%%%%%%%%%%%%%%%%%%%%%%%%%%%%%%%

% \newpage
%\section{Implementation, experiments [4 pages][Zafer]}

\section{Evaluation}
\label{sec:evaluation}
%  \iffalse
% In this section, we present our implementation of the instrumentation for aggregation functions, and evaluate how it performs on a set of example programs in comparison with state-of-the-art verification tools. \ze{Remove paragraph and start directly from implementation?}\pr{probably yes ...}
% \fi

To evaluate our instrumentation framework, we have implemented the instrumentation operators for quantifiers and aggregation over arrays. We evaluate this implementation on benchmarks taken from SV-COMP, as well as a set of C~programs by ourselves that use 
\ifuseimpl
\lstinline!\min!, \lstinline!\max!, \lstinline!\sum!, \lstinline!\forall!, \lstinline!\exists!, \lstinline!\product! and \lstinline!\numof!. 
\else
\lstinline!\min!, \lstinline!\max!, \lstinline!\sum!, and \lstinline!\forall!.
\fi

\subsection{Implementation}
 Our implementation is done in the setting of constrained Horn clauses (CHCs), by adding the rewrite rules defined in \autoref{sec:instrumentation_operators_for_arrays} to \eldarica~\cite{eldarica}, an open-source solver for CHCs.
We also implemented the automatic application of the instrumentation operators, largely following \autoref{alg:instrumentation_search}, but with a few minor changes due to the CHC setting. The CHC setting makes our implementation available to various CHC-based verification tools, for instance \jayhorn\ (Java) \cite{jayhorn-2017}, \korn\ (C)~\cite{DBLP:conf/tacas/Ernst23}, \rusthorn\ (Rust)~\cite{rusthorn}, \seahorn\ (C/LLVM)~\cite{seahorn}, and \tricera\ (C)~\cite{tricera-fmcad}.

%In order to make our implementation as available as possible, we chose to implement it over an intermediate language: constrained Horn clauses (CHCs). CHCs are generated by many verification tools targeting different programming languages, such as \tricera~(C)~\cite{tricera-fmcad}, \jayhorn~(Java)~\cite{jayhorn-2017}, \seahorn~(LLVM-based languages)~\cite{seahorn} and \rusthorn~(Rust)~\cite{rusthorn}.

%To implement our approach over CHCs, we extend \eldarica~\cite{eldarica}, an open-source model checker for CHCs, with a \emph{theory of aggregations} that allows the use of aggregation operators in its input CHCs. 
In order to evaluate our approach at the level of C programs, we extended \tricera, an open-source assertion-based model checker that translates C programs into a set of CHCs and relies on \eldarica\ as back-end solver. \tricera\ is extended to parse quantifiers and aggregation operators in C programs and to encode them as part of the translation into CHCs. We call the resulting tool chain \monocera. \ifuseimpl The implementation is based on our previous work~\cite{amilon-et-al-cav23}, but has been updated with, e.g., a better interface for defining new instrumentation operators. \fi 
An artefact that includes \monocera\ and the benchmarks is available online~\cite{artifact}\ifuseimpl \footnote{Note that this artefact contains the benchmarks and implementation of our previous paper~\cite{amilon-et-al-cav23}. In case of acceptance, we plan to update with a link to the new version of the implementation and the new benchmarks.}\fi.

%We have extended \eldarica\  \tricera\ is only extended 

\iffalse
The implementation in \eldarica\ can be summarised as follows. 
%Let $\mathit{Cl}$ be a set of CHCs encoding a program, and let $\Omega_f$ be the instrumentation (overloaded to work over CHCs) for some aggregate $f$ (such as \lstinline{\max} and \lstinline{\sum}). A set of instrumentations $R_{f_i} = \Omega_f(\mathit{CL})$ is computed for the $i$the application of $f$. 
Let $P_\mathit{Cl}$ be a set of CHCs encoding a program $P$.
\eldarica\ applies the instrumentation operator $\Omega$ to each aggregation function it detects in $P_\mathit{Cl}$ to produce a set 
of instrumentations $R$, comprising the instrumentation search space.
The instrumentation search given in \autoref{alg:instrumentation_search} is then started over $R$ using $P_\mathit{Cl}$. \marginpar{ZE: This is over-simplified, but going into more detail would require space and being much more precise...}

%comprising the instrumentation search space. The instrumentation search given in Algorithm~\ref{alg:instrumentation_search} is then applied over this search space. 
\fi
%The implementation supports any number of occurrences of \lstinline{\max}, \lstinline{\min}, \lstinline{\sum} and \lstinline{\product}. 
To handle complicated access patterns, for instance a program processing an array simultaneously from the beginning and from the end, the implementation can simultaneously apply multiple instrumentation operators; the number of operators
is incremented when \autoref{alg:instrumentation_search} returns \emph{Inconclusive}.
%We have also added support for tracking several ranges for each array

%, which allows us to verify, e.g., programs where an array is accessed from two ends simultaneously. 
%The workflow of the implementation is illustrated in Figure~\ref{fig:control_flow}.

% \subsection{Implementation}
% Our implementation was done in two parts. First, we performed the instrumentation on Horn-clause level, as part of the Eldarica model checker. Thereafter, we added syntactical support for extended quantifier in TriCera (which interfaces Eldarica as backend solver). The overall verification process for our implementation can thus be summarized as the following steps:

% \begin{enumerate}
%     \item Take program with assertions containing extended quantifiers.
%     \item Translate program (including the extended quantifiers) into Horn clauses.
%     \item Instrument the Horn-Clause representation of the program according to our approach.
%     \item Verify the program.
% \end{enumerate} 
% \marginpar{TODO: Add strategies into this descrption? Possibly do a workflow graph}

\subsection{Experiments and Comparisons}
To assess our implementation, we assembled a test suite and carried out experiments comparing \monocera\ with the state-of-the-art C~model checkers \cpachecker~2.1.1~\cite{cpachecker}, \seahorn~10.0.0~\cite{seahorn} and \tricera~0.2.
It should be noted that deductive verification frameworks, such as Dafny and Frama-C, can handle, for example, the program in \autoref{fig:quantified_example}, if they are provided with a manually written loop invariant. Since \monocera{} relies on automatic techniques for invariant inference, we only benchmark against tools using similar automatic techniques. We also excluded \veriabs~\cite{DBLP:conf/tacas/AfzalCCCDGKMUV20}, since its licence does not permit its use for scientific evaluation. 

The tools were set up, as far as possible, with equivalent configurations; for instance, to use the SMT-LIB theory of arrays~\cite{smtlib} in order to model C arrays, and a mathematical (as opposed to machine-based) semantics of integers. \cpachecker\ was configured to use $k$-induction~\cite{DBLP:conf/cav/0001DW15}, which was the only configuration that worked in our tests using mathematical integers. %(\cpachecker's default solver could not provide interpolants when using predicate abstraction).
\seahorn\ was run using the default settings. 
\ifuseimpl
All tests were run on a Linux machine with  Intel Xeon CPU E5-2630 v4 @ 2.20GHz and a timeout of 900 seconds.
\else
All tests were run on a Linux machine 
with AMD Opteron 2220 SE @ 2.8 GHz and 6 GB RAM with a timeout of 300 seconds.
\fi
\paragraph{Test Suite}
The comparison includes a set of programs calculating properties related to the quantification and aggregation properties over arrays. The benchmarks and verification results are summarised in \autoref{tab:results_monocera}. 
%where we also include comparisons with state-of-the-art tools \seahorn, \cpachecker\\marginpar{JA: Formatting and ref, ZE: SeaHorn have not been competing in SV-COMP for quite some time}, and standard \tricera. 
The benchmark suite contains programs ranging from 16 to 117 LOC, and is comprised of two parts:
\begin{inparaenum}[(i)]
    \item 117~programs taken from the SV-COMP repository~\cite{sv-benchmarks-2022}, and
    \ifuseimpl
    \item 35~programs crafted by the authors  (\lstinline!\min!: 6, \lstinline!\max!: 8, \lstinline!\sum!: 9, \lstinline!\forall!: 3, \lstinline!\exists!: 3,  \lstinline!\product!: 3, \lstinline!\numof!: 3).
    \else
    \item 26~programs crafted by the authors  (\lstinline!\min!: 6, \lstinline!\max!: 8, \lstinline!\sum!: 9, \lstinline!\forall!: 3).
    \fi
\end{inparaenum} 
\ifuseimpl The suite is based on our previous work~\cite{amilon-et-al-cav23}, but extended with (crafted) benchmarks for \lstinline!\exists!, \lstinline!\product! and \lstinline!\numof!.\fi

To construct the SV-COMP benchmark set for \monocera, we gathered all test files from the directories prefixed with \lstinline$array$ or \lstinline$loop$, and singled out programs containing some  assert statement that could be rewritten using a quantifier or an aggregation operator over a single array. For example, loops of the shape:
\begin{lstlisting}[style=ExtWhileStyleOp, basicstyle=\ttfamily\small,]
for (int i = 0; i < N; i++)
  assert(a[i] <= 0)
\end{lstlisting}
can be rewritten using \lstinline!\forall! or \lstinline!\max!.
We created a benchmark for each possible rewriting; for instance,
in the case of \lstinline!\max!, by rewriting the loop into
\begin{lstlisting}[style=ExtWhileStyleOp]
assert(\max(a, 0, N) <= 0)    
\end{lstlisting}
%
%That is, when evaluating \monocera, we rewrote such loops into a single assertion in terms of the appropriate operator.
%In particular, the loop in~\eqref{eq:loop_assert_max} can be rewritten into \lstinline[basicstyle={\small\ttfamily}]!assert(\max(a, 0, N) <= 0)!. 
The original benchmarks were used for the evaluation of the other tools, none of which supported (extended) quantifiers.
%
%Note also that some test programs occur in several categories, since they could be interpreted in terms of different operators. For example, the loop in~\eqref{eq:loop_assert_max} can be rewritten using either \texttt{max} or \texttt{forall}.
%

In (ii), we crafted 12 programs that make use of aggregation or quantifiers, and derived further benchmarks by considering different array sizes (10, 100 and unbounded size). One combination (unbounded array inside a struct) had to be excluded, as it is not valid C. In order to evaluate other tools on our crafted benchmarks, we reversed the process described for the SV-COMP benchmarks and translated the operators into corresponding loop constructs.

\iftrue
\begin{table}[tb]
\centering
\begin{minipage}{0.87\textwidth}
\footnotesize
\setlength{\belowrulesep }{0pt}
\setlength{\aboverulesep }{0pt}
%  \begin{center}
    \begin{tabular}{
    @{}
    l@{\hspace{0pt}}
    c@{\hspace{6pt}}
    c
    c
    c
    cc@{\hspace{12pt}}
    c
    c
    c
    c@{\hspace{8pt}}
    c@{\hspace{2pt}}
    c
    c@{\hspace{8pt}}
    c
    c
    }
& \multicolumn{5}{c}{\textbf{Verification results}}&&
\multicolumn{3}{@{\hspace{0pt}}c}{\textbf{Ver. time}} && \multicolumn{2}{@{\hspace{-2pt}}c}{\textbf{Inst. space}} && \multicolumn{2}{@{\hspace{-4pt}}c}{\textbf{Inst. steps}}
\\ \toprule
&  {\#Tests} & \monoc & \tri & \sea & \cpa  && Min & Max & Avg && Max & Avg && Max & Avg    
\\ \cmidrule(r{4pt}){2-2}\cmidrule(l{0pt}){3-6}\cmidrule{8-10}\cmidrule(r{3pt}){12-13}\cmidrule(r{3pt}){15-16}%\cmidrule(lr){2-2}\cmidrule(lr){3-3}\cmidrule(lr){5-5}\cmidrule(lr){9-11}
\verb!\min!     & 17  & 9  & 2 & 2 & 2 && 22 & 59  &  33 && 27 & 11 && 55 & 24 
\\
\verb!\max!  & 12   & 8 & 2 & 3 & 3  && 21  & 285 & 76 && 108 & 21 && 96 & 30 
\\
\verb!\sum!  & 26  & 16 & 3 & 3 & 3 && 26   & 245  & 78 && 2916 & 188 && 284 & 36 
\\
\verb!\forall! & 96 & 30 & 1 & 0 & 2 && 14  & 236 & 91 && 59049 & 2446 && 334 & 59 \\
%\texttt{crafted} & 30  & 24 & 1 & 1 & 1 && 21 & 245 & 74 && 2916 & 127 && 284 & 31 \\
%\hline
%Total          & xxx  & xx && xx & xxx &  xx   \\ 
\bottomrule
    \end{tabular}
  %  \end{center}
    \caption{Results for \monocera\ (\monoc), \tricera\ (\tri), \seahorn\ (\sea), and \cpachecker\ (\cpa). For \monocera, also statistics are given for verification time (s), size of the instrumentation search space, and search iterations.}
    \label{tab:results_monocera}
    \vspace*{-3ex}
\end{minipage}
\end{table}
\fi

\paragraph{Results}
In \autoref{tab:results_monocera}, we present the number of verified programs per instrumentation operator for each tool, as well as further statistics for \monocera\ regarding verification times and instrumentation search space. The ``Inst. space'' column indicates the size of the instrumentation search space (i.e., number of instrumentations producible by applying the non-deterministic instrumentation operator). The ``Inst. steps'' column indicates the number of attempted instrumentations, i.e., number of iterations in the while-loop in \autoref{alg:instrumentation_search}. In our implementation, the check in \autoref{alg:instrumentation_search}, line~5, can time out and cause the check to be repeated at a later time with a greater timeout, which can lead to more iterations than the size of the search space.
In 
Appendix~\ref{app:evaluation_results}, 
we list results per benchmark for each tool.

For the SV-COMP benchmarks, \cpachecker\ managed to verify 1 program, while \seahorn\ and \tricera\ could not verify any programs. \monocera\ verified in total 42 programs from SV-COMP. 
Regarding the crafted benchmarks, several tools could verify the examples with array size 10. However, when the array size was 100 or unbounded, only \monocera~succeeded. This is because only \monocera\ could infer the quantified invariants needed to prove unbounded safety  of those benchmarks.

\section{Related Work}

\iffalse
\begin{itemize}
    \item Methods to reason about recursively defined functions on algebraic data-types (catamorphisms, etc.): \cite{DBLP:journals/pacmpl/KSG22,DBLP:journals/tplp/AngelisPFP22}
    \item Derivation of quantified invariants over arrays in model checking: \cite{DBLP:conf/atva/GurfinkelSV18}
    \item Implicit representation of quantified array invariants in Horn solving: \cite{DBLP:conf/sas/BjornerMR13,DBLP:conf/sas/MonniauxG16}
    \item Handling comprehensions on arrays, sets, etc.\ deductive verification: \cite{DBLP:conf/sac/LeinoM09}
    \item Array folds logic, which can probably represent some of the extended quantifiers we are considering, and provides a decidable logic: \cite{DBLP:conf/cav/DacaHK16}
    \item Extended quantifiers~\cite{acsl}
    \item Rapid tool by Laura Kovacs
    \item \cite{DBLP:conf/cav/FedyukovichPMG19}
    \item \cite{DBLP:conf/cav/0001LMN13}
\end{itemize}
\fi

It is common practice, in both model checking and deductive verification, to translate high-level specifications to low-level specifications prior to verification (e.g., \cite{DBLP:reference/mc/2018,cell2010,DBLP:conf/sas/BjornerMR13,DBLP:conf/sas/MonniauxG16}). Such translations often make use of ghost variables and ghost code, although relatively little systematic research has been done on the required properties of ghost code~\cite{DBLP:journals/fmsd/FilliatreGP16}. The addition of ghost variables to a program for tracking the value of complex expressions also has similarities with the concept of term abstraction in Horn solving~\cite{DBLP:conf/lpar/AlbertiBGRS12}. To the best of our knowledge, we are presenting the first general framework for automatic program instrumentation.

A lot of research in \emph{software model checking} considered the handling of standard quantifiers~$\forall, \exists$ over arrays. In the setting of constrained Horn clauses, properties with universal quantifiers can sometimes be reduced to quantifier-free reasoning over non-linear Horn clauses~\cite{DBLP:conf/sas/BjornerMR13,DBLP:conf/sas/MonniauxG16}. Our approach follows the same philosophy of applying an up-front program transformation, but in a more general setting.
Various direct approaches to infer quantified array invariants have been proposed as well: e.g., by extending the IC3 algorithm~\cite{DBLP:conf/atva/GurfinkelSV18}, syntax-guided synthesis~\cite{DBLP:conf/cav/FedyukovichPMG19}, learning~\cite{DBLP:conf/cav/0001LMN13}, by solving recurrence equations~\cite{DBLP:conf/lpar/HenzingerHKR10},
backward reachability~\cite{DBLP:conf/lpar/AlbertiBGRS12}, or superposition~\cite{DBLP:conf/fmcad/GeorgiouGK20}. To the best of our knowledge, such methods have not been extended to aggregation.

\emph{Deductive verification} tools usually have rich support for quantified specifications, but rely on auxiliary assertions like loop invariants provided by the user, and on SMT solvers or automated theorem provers for quantifier reasoning.
Although several deductive verification tools can parse extended quantifiers, few offer support for reasoning about them. Our work is closest to the method for handling comprehension operators in Spec$\#$~\cite{DBLP:conf/sac/LeinoM09}, which relies on code annotations provided by the user, but provides heuristics to automatically verify such annotations. The code instrumentation presented in this article has similarity with the proof rules in Spec$\#$; the main differences are that our method is based on an upfront program transformation, and that we aim at automatically finding required program invariants, as opposed to only verifying their correctness. The KeY tool provides proof rules similar to the ones in Spec$\#$ for some of the JML extended quantifiers~\cite{DBLP:series/lncs/10001}; those proof rules can be applied manually to verify human-written invariants. The Frama-C system~\cite{DBLP:conf/sefm/CuoqKKPSY12} can parse ACSL extended quantifiers~\cite{acsl}, but, to the best of our knowledge, none of the Frama-C plugins can automatically process such quantifiers. Other systems, e.g., Dafny~\cite{dafny}, require users to manually define aggregation operators as recursive functions.

In the theory of \emph{algebraic data-types}, several transformation-based approaches have been proposed to verify properties that involve recursive functions or catamorphisms~\cite{DBLP:journals/pacmpl/KSG22,DBLP:journals/tplp/AngelisPFP22}. Aggregation over arrays resembles the evaluation of recursive functions over data-types; a major difference is that data-types are more restricted with respect to accessing and updating data than arrays.

Array folds logic (AFL)~\cite{DBLP:conf/cav/DacaHK16} is a decidable logic in which properties on arrays beyond standard quantification can be expressed: for instance, counting the number of elements with some property. Similar properties can be expressed using automata on data words~\cite{DBLP:conf/csl/Segoufin06}, or in variants of monadic second-order logic~\cite{DBLP:journals/tocl/NevenSV04}. Such languages can be seen as alternative formalisms to aggregation or extended quantifiers; they do not cover, however, all kinds of aggregation we are interested in. Array sums cannot be expressed in AFL or data automata, for instance.

\section{Conclusion}
\label{sec:conclusion}

\iffalse
Future work:
\begin{itemize}
    \item Extraction of invariants/annotations with extended quantifiers
    \item Combining operators? (E.g., Filter/forall-Sum)
    \item Adding more operators
    \item Improving search algorithm
    \item Stronger completeness results
\end{itemize}
\fi

% We have presented a new approach for the challenging problem of verifying
% programs with aggregation, and obtained encouraging results in an evaluation
% on programs taken from the SV-COMP. The program transformation framework we introduced is general,
% and various avenues exist that we plan to explore in future work, such as the use of our
% approach in the context of deductive verification, the application
% of the framework to other challenging specification operators, and the re-introduction
% of aggregation in inferred loop invariants or contracts.

We have presented a framework for automatic and provably correct program instrumentation, allowing the automatic verification of programs containing certain expressive language constructs, which are not directly supported by the existing automatic verification tools. 
Our experiments with a prototypical implementation, in the tool \monocera, show that our method is able to automatically verify a significant number of benchmark programs involving quantification and aggregation over arrays that are beyond the scope of other tools.

There are still various other benchmarks that \monocera\ (as well as other tools) cannot verify. We believe that many of those benchmarks are in reach of our method, because of the generality of our approach.
Ghost code is known to be a powerful
specification mechanism; similarly, in our setting, more powerful instrumentation operators can be easily formulated for specific kinds of programs.
In future work, we therefore plan to develop a library of instrumentation operators for different language constructs  (including arithmetic operators), non-linear arithmetic,
other types of structures with regular access patterns
such as binary heaps,
and general linked-data structures.
We also plan to refine our method for showing incorrectness of programs more efficiently, as the approach is currently applicable mainly for verifying correctness.
%\iftr\else\ (experiments in~\cite{techreport})\fi.
% Another line of work is the establishment of stronger completeness results than the weak completeness result presented here, for specific programming language fragments. 
% \todo[inline]{MV: Change conclusion since we have actually done some of the future work discussed in this paper!}

\iffalse
% To be added later, for the final version
\paragraph{Acknowledgements}
This work has been partially funded by the Swedish  Vinnova FFI Programme under grant 2021-02519, the Swedish Research Council (VR)
    under grant~2018-04727, the Swedish Foundation for Strategic
    Research (SSF) under the project WebSec (Ref.\ RIT17-0011), and the
    Wallenberg project UPDATE. We are also grateful for the opportunity
    to discuss the research at the Dagstuhl Seminar~22451 on ``Principles of Contract Languages.''
\fi

%%%%%%%%%%%%%%%%%%%%%%%%%%%%%%%%%%%%%%%%%%%%%%%%%%%%%%%%%%%%%%%%
%%%%%%%%%%%%%%%%%%%%%%%%%%%%%%%%%%%%%%%%%%%%%%%%%%%%%%%%%%%%%%%%

% \clearpage
% \bibliography{references}

\bibliographystyle{plain}
\bibliography{references}% common bib file

\newpage
\begin{appendices}

%%%%%%%%%%%%%%%%%%%%%%%%%%%%%%%%%%%%%%%%%

\section{Detailed Evaluation Results}
\label{app:evaluation_results}
We list all evaluated benchmarks in \autoref{tbl:per-benchmark-results-forall} (\emph{forall}), \autoref{tbl:per-benchmark-results-max} (\emph{max}), \autoref{tbl:per-benchmark-results-min} (\emph{min}) and \autoref{tbl:per-benchmark-results-sum} (\emph{sum}) for all evaluated tools: \monocera\ (\monoc), \tricera\ (\tri), \seahorn\ (\sea), and \cpachecker\ (\cpa). The list includes mostly SV-COMP benchmarks and the first column indicates their original directories (should be prepended with ``\texttt{c/}''). There are some crafted benchmarks that are not from SV-COMP; for these we only state the name of the benchmark (for instance \texttt{forall1-10}).

The expected result for all benchmarks in the tables is \textbf{\textcolor{green!50!black}{True}}. Durations are given next to the returned result for benchmarks which did not time out. 

On top of quantifying over arrays, many of the benchmarks have assertions over nonlinear integer arithmetic (NIA). Decision procedures for NIA are incomplete, and we have seen errors in the presence of such assertions in some of the evaluated tools.

The result ``Unknown'' is returned by \monocera\ when the instrumentation search space is exhausted.

{
\scriptsize
\renewcommand{\arraystretch}{.5}

% \begin{table}
  \begin{longtable}{lrrrr}
  \caption{\emph{forall} benchmark results for all tools. Timeout (T/O)
  is 300.0 s.}
   \label{tbl:per-benchmark-results-forall}\\
      & \cpa & \sea & \tri & \monoc\\\toprule
 \endfirsthead
      & \cpa & \sea & \tri & \monoc\\\toprule
 \endhead
     array-cav19/array\_doub\_access\_init\_const.c & T/O & T/O & T/O & \textbf{\textcolor{green!50!black}{True}} (111) \\\midrule
		array-cav19/array\_init\_nondet\_vars.c & T/O & T/O & T/O & \textbf{\textcolor{green!50!black}{True}} (14) \\\midrule
		array-cav19/array\_init\_pair\_sum\_const.c & T/O & T/O & T/O & T/O \\\midrule
		array-cav19/array\_init\_pair\_symmetr.c & T/O & T/O & T/O & T/O \\\midrule
		array-cav19/array\_init\_pair\_symmetr2.c & T/O & T/O & T/O & T/O \\\midrule
		array-cav19/array\_init\_var\_plus\_ind.c & T/O & T/O & Error (2) & \textbf{\textcolor{green!50!black}{True}} (23) \\\midrule
		array-cav19/array\_init\_var\_plus\_ind2.c & T/O & T/O & Error (2) & Unknown (109) \\\midrule
		array-cav19/array\_init\_var\_plus\_ind3.c & T/O & T/O & Error (3) & \textbf{\textcolor{green!50!black}{True}} (60) \\\midrule
		array-industry-pattern/array\_ptr\_partial\_init.c & T/O & T/O & T/O & T/O \\\midrule
		array-industry-pattern/array\_shadowinit.c & T/O & T/O & T/O & \textbf{\textcolor{green!50!black}{True}} (19) \\\midrule
		array-cav19/array\_tiling\_poly6.c & T/O & Error (0) & T/O & T/O \\\midrule
		array-cav19/array\_tripl\_access\_init\_const.c & T/O & T/O & T/O & \textbf{\textcolor{green!50!black}{True}} (199) \\\midrule
		array-fpi/condg.c & T/O & T/O & T/O & T/O \\\midrule
		array-fpi/condm.c & T/O & T/O & T/O & Error (8) \\\midrule
		array-fpi/condn.c & T/O & T/O & T/O & \textbf{\textcolor{green!50!black}{True}} (88) \\\midrule
		array-fpi/eqn1.c & T/O & Error (0) & T/O & T/O \\\midrule
		array-fpi/eqn2.c & T/O & Error (0) & T/O & T/O \\\midrule
		array-fpi/eqn3.c & T/O & Error (0) & T/O & T/O \\\midrule
		array-fpi/eqn4.c & T/O & Error (0) & T/O & T/O \\\midrule
		array-fpi/eqn5.c & T/O & Error (0) & T/O & T/O \\\midrule
		forall1-10 & \textbf{\textcolor{green!50!black}{True}} (23) & Error (0) & \textbf{\textcolor{green!50!black}{True}} (55) & \textbf{\textcolor{green!50!black}{True}} (37) \\\midrule
		forall1-100 & T/O & Error (2) & T/O & \textbf{\textcolor{green!50!black}{True}} (31) \\\midrule
		forall1-UB & T/O & Error (0) & T/O & \textbf{\textcolor{green!50!black}{True}} (32) \\\midrule
		array-fpi/ifcomp.c & T/O & Error (0) & T/O & T/O \\\midrule
		array-fpi/ifeqn1.c & T/O & Error (0) & T/O & T/O \\\midrule
		array-fpi/ifeqn2.c & T/O & Error (0) & T/O & T/O \\\midrule
		array-fpi/ifeqn3.c & T/O & Error (0) & T/O & T/O \\\midrule
		array-fpi/ifeqn4.c & T/O & Error (0) & T/O & T/O \\\midrule
		array-fpi/ifeqn5.c & T/O & Error (0) & T/O & T/O \\\midrule
		array-fpi/ifncomp.c & Error (20) & Error (0) & T/O & T/O \\\midrule
		array-fpi/indp1.c & T/O & Error (0) & T/O & T/O \\\midrule
		array-fpi/indp2.c & T/O & Error (0) & T/O & T/O \\\midrule
		array-fpi/indp3.c & T/O & Error (0) & T/O & T/O \\\midrule
		array-fpi/indp4.c & T/O & Error (0) & T/O & T/O \\\midrule
		array-tiling/mbpr2.c & T/O & T/O & T/O & T/O \\\midrule
		array-tiling/mbpr3.c & T/O & T/O & T/O & T/O \\\midrule
		array-tiling/mbpr4.c & T/O & T/O & T/O & T/O \\\midrule
		array-tiling/mbpr5.c & T/O & T/O & T/O & T/O \\\midrule
		loop-lit/mcmillan2006.c & \textbf{\textcolor{green!50!black}{True}} (7) & T/O & T/O & \textbf{\textcolor{green!50!black}{True}} (19) \\\midrule
		array-tiling/mlceu2.c & T/O & T/O & Error (3) & Error (4) \\\midrule
		array-fpi/modn.c & Error (20) & Error (0) & T/O & T/O \\\midrule
		array-fpi/modp.c & Error (17) & Error (0) & T/O & Error (9) \\\midrule
		array-fpi/mods.c & T/O & Error (0) & T/O & T/O \\\midrule
		array-fpi/ncomp.c & T/O & Error (0) & T/O & T/O \\\midrule
		array-tiling/nr2.c & T/O & T/O & T/O & \textbf{\textcolor{green!50!black}{True}} (93) \\\midrule
		array-tiling/nr3.c & T/O & T/O & T/O & \textbf{\textcolor{green!50!black}{True}} (123) \\\midrule
		array-tiling/nr4.c & T/O & T/O & T/O & \textbf{\textcolor{green!50!black}{True}} (156) \\\midrule
		array-tiling/nr5.c & T/O & T/O & T/O & \textbf{\textcolor{green!50!black}{True}} (223) \\\midrule
		array-fpi/nsqm.c & T/O & Error (0) & T/O & T/O \\\midrule
		array-fpi/nsqm-if.c & Error (15) & Error (0) & T/O & T/O \\\midrule
		array-lopstr16/partial\_lesser\_bound.c & T/O (108) & T/O & T/O & T/O \\\midrule
		array-lopstr16/partial\_lesser\_bound-1.c & T/O & Error (1) & Error (2) & \textbf{\textcolor{green!50!black}{True}} (89) \\\midrule
		array-fpi/pcomp.c & T/O & Error (0) & T/O & T/O \\\midrule
		array-tiling/pnr2.c & T/O & T/O & T/O & \textbf{\textcolor{green!50!black}{True}} (227) \\\midrule
		array-tiling/pnr3.c & T/O & T/O & T/O & \textbf{\textcolor{green!50!black}{True}} (236) \\\midrule
		array-tiling/pnr4.c & T/O & T/O & T/O & T/O \\\midrule
		array-tiling/pnr5.c & T/O & T/O & T/O & T/O \\\midrule
		array-tiling/poly1.c & T/O & Error (0) & T/O & \textbf{\textcolor{green!50!black}{True}} (25) \\\midrule
		array-tiling/poly2.c & T/O & Error (0) & T/O & T/O \\\midrule
		array-tiling/pr2.c & T/O & Error (0) & T/O & \textbf{\textcolor{green!50!black}{True}} (117) \\\midrule
		array-tiling/pr3.c & T/O & Error (0) & T/O & \textbf{\textcolor{green!50!black}{True}} (143) \\\midrule
		array-tiling/pr4.c & T/O & Error (0) & T/O & T/O \\\midrule
		array-tiling/pr5.c & T/O & Error (0) & T/O & T/O \\\midrule
		array-tiling/rew.c & T/O & T/O & T/O & T/O \\\midrule
		array-tiling/rewnif.c & T/O & T/O & T/O & \textbf{\textcolor{green!50!black}{True}} (55) \\\midrule
		array-tiling/rewnifrev.c & T/O & T/O & T/O & \textbf{\textcolor{green!50!black}{True}} (60) \\\midrule
		array-tiling/rewnifrev2.c & T/O & T/O & T/O & \textbf{\textcolor{green!50!black}{True}} (97) \\\midrule
		array-examples/sanfoundry\_27\_ground.c & T/O & T/O & T/O & Error (19) \\\midrule
		loop-crafted/simple\_array\_index\_value\_2.c & T/O & Error (1) & Error (2) & \textbf{\textcolor{green!50!black}{True}} (39) \\\midrule
		loop-crafted/simple\_array\_index\_value\_3.c & T/O & Error (1) & Error (3) & \textbf{\textcolor{green!50!black}{True}} (33) \\\midrule
		array-fpi/sina1.c & T/O & T/O & T/O & \textbf{\textcolor{green!50!black}{True}} (199) \\\midrule
		array-fpi/sina2.c & T/O & T/O & T/O & T/O \\\midrule
		array-fpi/sina3.c & T/O & T/O & T/O & T/O \\\midrule
		array-fpi/sina4.c & T/O & T/O & T/O & T/O \\\midrule
		array-fpi/sina5.c & T/O & T/O & T/O & T/O \\\midrule
		array-tiling/skipped.c & T/O & T/O & T/O & T/O \\\midrule
		array-fpi/sqm.c & T/O & Error (0) & T/O & T/O \\\midrule
		array-fpi/sqm-if.c & T/O & Error (0) & T/O & T/O \\\midrule
		array-fpi/ss2.c & Error (65) & Error (0) & T/O & T/O \\\midrule
		array-fpi/ssina.c & Error (38) & Error (0) & T/O & T/O \\\midrule
		array-examples/standard\_copyInitSum2\_ground-2.c & T/O & T/O & T/O & T/O \\\midrule
		array-examples/standard\_copyInitSum3\_ground.c & T/O & T/O & T/O & T/O \\\midrule
		array-examples/standard\_copyInit\_ground.c & T/O & T/O & T/O & T/O \\\midrule
		array-examples/standard\_init1\_ground-2.c & T/O & T/O & T/O & \textbf{\textcolor{green!50!black}{True}} (32) \\\midrule
		array-examples/standard\_init2\_ground-2.c & T/O & T/O & T/O & \textbf{\textcolor{green!50!black}{True}} (82) \\\midrule
		array-examples/standard\_init3\_ground-2.c & T/O & T/O & T/O & T/O \\\midrule
		array-examples/standard\_init4\_ground-2.c & T/O & T/O & T/O & T/O \\\midrule
		array-examples/standard\_init5\_ground-1.c & T/O & T/O & T/O & T/O \\\midrule
		array-examples/standard\_init6\_ground-2.c & T/O & T/O & T/O & T/O \\\midrule
		array-examples/standard\_init7\_ground-2.c & T/O & T/O & T/O & T/O \\\midrule
		array-examples/standard\_init8\_ground-2.c & T/O & T/O & T/O & T/O \\\midrule
		array-examples/standard\_init9\_ground-2.c & T/O & T/O & T/O & T/O \\\midrule
		array-examples/standard\_maxInArray\_ground.c & T/O & T/O & T/O & T/O \\\midrule
		array-examples/standard\_minInArray\_ground-2.c & T/O & T/O & T/O & T/O \\\midrule
		array-examples/standard\_partition\_ground-2.c & T/O & T/O & T/O & T/O \\\midrule
		array-examples/standard\_vararg\_ground.c & T/O & T/O & T/O & \textbf{\textcolor{green!50!black}{True}} (57)\\\bottomrule
  \end{longtable}
 % \end{table}

% \begin{table}
  \begin{longtable}{lrrrr}
  \caption{\emph{max} benchmark results for all tools. Timeout (T/O)
  is 300.0 s.}
   \label{tbl:per-benchmark-results-max}\\
      & \cpa & \sea & \tri & \monoc\\\toprule
 \endfirsthead
      & \cpa & \sea & \tri & \monoc\\\toprule
 \endhead
     array-cav19/array\_init\_var\_plus\_ind3.c & T/O & T/O & Error (2) & Unknown (277) \\\midrule
		battery\_diag-10 & \textbf{\textcolor{green!50!black}{True}} (155) & \textbf{\textcolor{green!50!black}{True}} (0) & T/O & T/O \\\midrule
		battery\_diag-100 & T/O & T/O & T/O & T/O \\\midrule
		array-fpi/condn.c & T/O & T/O & T/O & \textbf{\textcolor{green!50!black}{True}} (129) \\\midrule
		max\_eq-10 & \textbf{\textcolor{green!50!black}{True}} (21) & \textbf{\textcolor{green!50!black}{True}} (0) & \textbf{\textcolor{green!50!black}{True}} (78) & \textbf{\textcolor{green!50!black}{True}} (21) \\\midrule
		max\_eq-100 & T/O & T/O & T/O & \textbf{\textcolor{green!50!black}{True}} (38) \\\midrule
		max\_eq-UB & T/O & T/O & T/O & \textbf{\textcolor{green!50!black}{True}} (27) \\\midrule
		max\_leq-10 & \textbf{\textcolor{green!50!black}{True}} (25) & \textbf{\textcolor{green!50!black}{True}} (0) & \textbf{\textcolor{green!50!black}{True}} (90) & \textbf{\textcolor{green!50!black}{True}} (23) \\\midrule
		max\_leq-100 & T/O & T/O & T/O & \textbf{\textcolor{green!50!black}{True}} (30) \\\midrule
		max\_leq-UB & T/O & T/O & T/O & \textbf{\textcolor{green!50!black}{True}} (52) \\\midrule
		array-examples/sanfoundry\_27\_ground.c & T/O & T/O & T/O & \textbf{\textcolor{green!50!black}{True}} (285) \\\midrule
		array-examples/standard\_maxInArray\_ground.c & T/O & T/O & T/O & T/O\\\bottomrule
  \end{longtable}
 % \end{table}

% \begin{table}
  \begin{longtable}{lrrrr}
  \caption{\emph{min} benchmark results for all tools. Timeout (T/O)
  is 300.0 s.}
   \label{tbl:per-benchmark-results-min}\\
      & \cpa & \sea & \tri & \monoc\\\toprule
 \endfirsthead
      & \cpa & \sea & \tri & \monoc\\\toprule
 \endhead
     array-cav19/array\_doub\_access\_init\_const.c & T/O & T/O & T/O & T/O \\\midrule
		array-cav19/array\_init\_pair\_sum\_const.c & T/O & T/O & T/O & T/O \\\midrule
		array-cav19/array\_init\_pair\_symmetr2.c & T/O & T/O & T/O & T/O \\\midrule
		array-cav19/array\_init\_var\_plus\_ind.c & T/O & T/O & Error (3) & Unknown (67) \\\midrule
		array-cav19/array\_init\_var\_plus\_ind2.c & T/O & T/O & Error (2) & Unknown (64) \\\midrule
		array-cav19/array\_tripl\_access\_init\_const.c & T/O & T/O & T/O & T/O \\\midrule
		min\_eq-10 & \textbf{\textcolor{green!50!black}{True}} (25) & \textbf{\textcolor{green!50!black}{True}} (0) & \textbf{\textcolor{green!50!black}{True}} (61) & \textbf{\textcolor{green!50!black}{True}} (28) \\\midrule
		min\_eq-100 & T/O & T/O & T/O & \textbf{\textcolor{green!50!black}{True}} (22) \\\midrule
		min\_eq-UB & T/O & T/O & T/O & \textbf{\textcolor{green!50!black}{True}} (25) \\\midrule
		min\_geq-10 & \textbf{\textcolor{green!50!black}{True}} (26) & \textbf{\textcolor{green!50!black}{True}} (0) & \textbf{\textcolor{green!50!black}{True}} (51) & \textbf{\textcolor{green!50!black}{True}} (26) \\\midrule
		min\_geq-100 & T/O & T/O & T/O & \textbf{\textcolor{green!50!black}{True}} (28) \\\midrule
		min\_geq-UB & T/O & T/O & T/O & \textbf{\textcolor{green!50!black}{True}} (24) \\\midrule
		array-tiling/rew.c & T/O & T/O & T/O & \textbf{\textcolor{green!50!black}{True}} (40) \\\midrule
		array-tiling/rewnifrev.c & T/O & T/O & T/O & \textbf{\textcolor{green!50!black}{True}} (41) \\\midrule
		array-examples/standard\_minInArray\_ground-2.c & T/O & T/O & T/O & T/O \\\midrule
		array-examples/standard\_partition\_ground-2.c & T/O & T/O & T/O & T/O \\\midrule
		array-examples/standard\_vararg\_ground.c & T/O & T/O & T/O & \textbf{\textcolor{green!50!black}{True}} (59)\\\bottomrule
  \end{longtable}
 % \end{table}

% \begin{table}
  \begin{longtable}{lrrrr}
  \caption{\emph{sum} benchmark results for all tools. Timeout (T/O)
  is 300.0 s.}
   \label{tbl:per-benchmark-results-sum}\\
      & \cpa & \sea & \tri & \monoc\\\toprule
 \endfirsthead
      & \cpa & \sea & \tri & \monoc\\\toprule
 \endhead
     array-fpi/brs1.c & T/O & T/O & T/O & \textbf{\textcolor{green!50!black}{True}} (26) \\\midrule
		array-fpi/brs2.c & T/O & Error (0) & T/O & \textbf{\textcolor{green!50!black}{True}} (34) \\\midrule
		array-fpi/brs3.c & T/O & T/O & T/O & \textbf{\textcolor{green!50!black}{True}} (41) \\\midrule
		array-fpi/brs4.c & T/O & T/O & T/O & \textbf{\textcolor{green!50!black}{True}} (38) \\\midrule
		array-fpi/brs5.c & T/O & T/O & T/O & \textbf{\textcolor{green!50!black}{True}} (58) \\\midrule
		array-fpi/conda.c & T/O & T/O & T/O & T/O \\\midrule
		array-fpi/indp5.c & Error (17) & Error (0) & T/O & T/O \\\midrule
		array-fpi/ms1.c & T/O & T/O & T/O & \textbf{\textcolor{green!50!black}{True}} (26) \\\midrule
		array-fpi/ms2.c & T/O & T/O & T/O & \textbf{\textcolor{green!50!black}{True}} (46) \\\midrule
		array-fpi/ms3.c & T/O & T/O & T/O & \textbf{\textcolor{green!50!black}{True}} (31) \\\midrule
		array-fpi/ms4.c & T/O & T/O & T/O & \textbf{\textcolor{green!50!black}{True}} (33) \\\midrule
		array-fpi/ms5.c & T/O & T/O & T/O & \textbf{\textcolor{green!50!black}{True}} (32) \\\midrule
		array-fpi/s1lif.c & T/O & T/O & T/O & T/O \\\midrule
		array-fpi/s2lif.c & T/O & T/O & T/O & T/O \\\midrule
		array-fpi/s3lif.c & T/O & T/O & T/O & T/O \\\midrule
		array-fpi/s4lif.c & T/O & T/O & T/O & T/O \\\midrule
		array-fpi/s5lif.c & T/O & T/O & T/O & T/O \\\midrule
		sum\_eq-10 & \textbf{\textcolor{green!50!black}{True}} (50) & \textbf{\textcolor{green!50!black}{True}} (0) & \textbf{\textcolor{green!50!black}{True}} (45) & \textbf{\textcolor{green!50!black}{True}} (245) \\\midrule
		sum\_eq-100 & T/O & T/O & T/O & T/O \\\midrule
		sum\_eq-UB & T/O & T/O & T/O & \textbf{\textcolor{green!50!black}{True}} (55) \\\midrule
		sum\_geq-10 & \textbf{\textcolor{green!50!black}{True}} (45) & \textbf{\textcolor{green!50!black}{True}} (0) & \textbf{\textcolor{green!50!black}{True}} (98) & \textbf{\textcolor{green!50!black}{True}} (200) \\\midrule
		sum\_geq-100 & T/O & T/O & T/O & \textbf{\textcolor{green!50!black}{True}} (221) \\\midrule
		sum\_geq-UB & T/O & T/O & T/O & \textbf{\textcolor{green!50!black}{True}} (59) \\\midrule
		two\_arrays\_sum-10 & \textbf{\textcolor{green!50!black}{True}} (26) & \textbf{\textcolor{green!50!black}{True}} (0) & \textbf{\textcolor{green!50!black}{True}} (46) & T/O \\\midrule
		two\_arrays\_sum-100 & T/O & T/O & T/O & T/O \\\midrule
		two\_arrays\_sum-UB & T/O & T/O & T/O & \textbf{\textcolor{green!50!black}{True}} (106)\\\bottomrule
  \end{longtable}
 % \end{table}

}
\end{appendices}

\end{document}